\documentclass[11pt]{amsart}
\usepackage{hyperref} 
 \usepackage{amsmath}
\usepackage{amssymb}
\usepackage{algorithmic}
\usepackage{algorithm}
\usepackage{comment}
\usepackage{amsmath}
\usepackage{mathrsfs}
\usepackage{color}
\usepackage{enumerate}
 \usepackage{subfig}
 \usepackage{comment}
  \usepackage{algorithmic}

\renewcommand{\eqref}[1]{(\ref{#1})}

\newcommand{\A}{{\mathcal A}}

\newcommand{\vx}{{\boldsymbol x}}
\newcommand{\vh}{{\boldsymbol h}}
\newcommand{\vX}{{\boldsymbol X}}
\newcommand{\vW}{{\boldsymbol W}}
\newcommand{\vy}{{\boldsymbol y}}
\newcommand{\vY}{{\boldsymbol Y}}
\newcommand{\vA}{{\boldsymbol A}}
\newcommand{\vz}{{\boldsymbol z}}
\newcommand{\vw}{{\boldsymbol w}}

\newcommand{\va}{{\boldsymbol a}}
\newcommand{\vb}{{\boldsymbol b}}
\newcommand{\vm}{{\boldsymbol m}}
\newcommand{\vd}{{\boldsymbol d}}
\newcommand{\vv}{{\boldsymbol v}}
\newcommand{\vH}{{\boldsymbol H}}
\newcommand{\vD}{{\boldsymbol D}}
\newcommand{\vC}{{\boldsymbol C}}

\newcommand{\vu}{{\boldsymbol u}}
\newcommand{\vf}{{\boldsymbol f}}
\newcommand{\vZ}{{\boldsymbol Z}}

\newcommand{\vU}{{\boldsymbol U}}
\newcommand{\vV}{{\boldsymbol V}}

\newcommand{\ve}{{\boldsymbol e}}

\newcommand{\vF}{{\boldsymbol F}}

\allowdisplaybreaks[4]

\usepackage[T1]{fontenc}
\usepackage[latin9]{inputenc}
\usepackage{geometry}
\usepackage{amsmath}
\usepackage{amsthm}
\usepackage{amssymb}
\usepackage{graphicx}
 \usepackage{subfig}
\theoremstyle{plain}

\newtheorem{prop}{Proposition}[section]
\newtheorem{lem}[prop]{Lemma}
\newtheorem{thm}{Theorem}[section]
\newtheorem{remark}{Remark}[section]

\newtheorem{defn}{Definition}[section]
\newtheorem{coro}{Corollary}[section]

\ifx\proof\undefined
\newenvironment{proof}[1][\protect\proofname]{\par
\normalfont\topsep6\p@\@plus6\p@\relax
\trivlist
\itemindent\parindent
\item[\hskip\labelsep\scshape #1]\ignorespaces
}{%
\endtrivlist\@endpefalse
}
\providecommand{\proofname}{Proof}
\fi
\numberwithin{equation}{section}
\makeatother

\usepackage{babel}

\begin{document}
\title{Stable Recovery Guarantees for Blind Deconvolution under Random Mask Assumption}

\author{Song Li}
\thanks{ Song Li is supported  by NSFC grant (U21A20426, 12071426)}
\address{School of Mathematical Sciences, Zhejiang University, 38 Zheda Road, Hangzhou, 310027, China}
\email{songli@zju.edu.cn}
\author{Yu Xia}
\thanks{ Yu Xia was supported by NSFC grant (12271133, 11901143) and the key project of Zhejiang Provincial Natural Science Foundation grant (LZ23A010002)}
\address{School of Mathematics, Hangzhou Normal University, Hangzhou 311121, China}
\email{yxia@hznu.edu.cn}



\maketitle

\begin{abstract}
This study addresses the blind deconvolution problem with modulated inputs, focusing on a measurement model where an unknown blurring kernel $\boldsymbol{h}$ is convolved with multiple random modulations $\{\boldsymbol{d}_l\}_{l=1}^{L}$(coded masks) of a signal $\boldsymbol{x}$, subject to $\ell_2$-bounded noise. We introduce a more generalized framework for coded masks, enhancing the versatility of our approach. Our work begins within a constrained least squares framework, where we establish a robust recovery bound for both $\boldsymbol{h}$ and $\boldsymbol{x}$, demonstrating its near-optimality up to a logarithmic factor. Additionally, we present a new recovery scheme that leverages sparsity constraints on $\boldsymbol{x}$. This approach significantly reduces the sampling complexity to the order of $L=O(\log n)$ when the non-zero elements of $\boldsymbol{x}$ are sufficiently separated. Furthermore, we demonstrate that incorporating sparsity constraints yields a refined error bound compared to the traditional constrained least squares model. The proposed method results in more robust and precise signal recovery, as evidenced by both theoretical analysis and numerical simulations. These findings contribute to advancing the field of blind deconvolution and offer potential improvements in various applications requiring signal reconstruction from modulated inputs.
\end{abstract}

\section{Introduction}
\subsection{Problem Setup}\label{sec: problems}
Blind deconvolution is an inverse problem that aims to reconstruct two unknown signals, $\boldsymbol{h}, \boldsymbol{x} \in \mathbb{C}^n$, from their circular convolution $\boldsymbol{y}\in \mathbb{C}^n$, defined as $\vy: = \boldsymbol{h} \circledast \boldsymbol{x}$, where $\circledast$ denotes the circular convolution operator. This operation can be equivalently expressed in matrix form as $\boldsymbol{y}=\boldsymbol{h}\circledast\boldsymbol{x}=\vC_{\boldsymbol{h}}\boldsymbol{x}$, where $\vC_{\boldsymbol{h}}$  is the circulant matrix generated by $\boldsymbol{h}=[{h}_{1},\cdots,{h}_{n}]^{T}$, defined as:
 \[
 \vC_{\boldsymbol{h}}=\begin{bmatrix}h_{1} & h_{n} & \cdots & h_{2}\\
h_{2} & h_{1} & \cdots & h_{3}\\
\vdots & \vdots & \ddots & \vdots\\
h_{n} & h_{n-1} & \cdots & h_{1}
\end{bmatrix}.
\]
 This problem arises in numerous fields, including astronomy, optics, image processing, and communications engineering \cite{JC93, CW98, LWDF11, WBSJ15}. The blind deconvolution problem is inherently ill-posed due to the presence of scaled-shift symmetry, which implies that there are infinitely many signal pairs that can yield the same convolution result. Consequently, incorporating prior information is crucial to overcoming this ill-posedness. For example, one might impose a subspace condition on $\boldsymbol{x}$ \cite{ARR14}, or enforce a short support condition on $\boldsymbol{h}$ in combination with a sparsity constraint on $\boldsymbol{x}$ \cite{SAS, Ji}.
 
 In this work, we examine a related class of blind deconvolution problems, where the blur kernel ${\vh} \in \mathbb{C}^{n}$ is convolved with multiple modulated inputs. Specifically, the observations $\vy_l\in \mathbb{C}^{n}$, $l=1,\ldots,L$,  take the form
\begin{equation}
\label{eq: original measurement}
{\boldsymbol{y}}_{l}:=\boldsymbol{h}\circledast\left(\boldsymbol{d}_{l}\odot\boldsymbol{x}\right),\qquad\text{for}\  l=1,\cdots,L,
\end{equation}
where $\odot$ denotes the Hadamard product (element-wise multiplication), defined for vectors $\va=[a_1,\ldots,a_n]^{T}$  and $\vb=[b_1,\ldots,b_n]^{T}$ as $\va\odot\vb=[a_1b_1,\ldots,a_nb_n]^{T}$. This formulation can also be interpreted as a self-calibration problem \cite{nonconvex_random_calibration, BD_LS}. Here, $\boldsymbol{d}_l\in \mathbb{C}^{n}$, for $l = 1, \dots, L$, are known coded masks, and the objective is to recover both $\boldsymbol{h}$ and $\boldsymbol{x}$ from the observations $\boldsymbol{y}_{l}$, using the smallest possible number of measurements $L$.

Such modulations can be practically implemented via optical diffraction gratings \cite{diffractiongrating}. Prior theoretical analyses predominantly focused on Rademacher-distributed random masks, where the elements of $\boldsymbol{d}_l$, $l=1,\ldots,L$, are independently sampled from the discrete uniform distribution on $\{ \pm 1 \}$ \cite{BD_randomSign, Romberg_RandomMask, BD_LS}. 

However, in practical applications, a diverse spectrum of mask configurations is frequently employed. For deconvolution problems, particularly in motion deblurring, rapid modulation of camera exposure using broad-band coded apertures is widely adopted. These apertures often utilize  Walsh-Hadamard codes and Modified Uniformly Redundant Arrays (MURA) codes \cite{motion_blur_coded}.
Furthermore, complex-valued random masks, exemplified by phase masks characterized by the complex exponential $\exp(\mathrm{i}\theta)$ where $\theta\in[0,2\pi)$, are extensively utilized in image deconvolution problems \cite{phasemask}. The theoretical foundation for these phase masks is rooted in the concept of Orbital Angular Momentum (OAM) of electromagnetic waves \cite{physics}.
Beyond deblurring applications, both real-valued and complex-valued masks find significant applications in adjacent areas. Notable among these are phase retrieval problems \cite{Candes15, CDP, Chen15}, where the mask designs play a crucial role in the reconstruction of phase information from intensity measurements.

Motivated by this broader context, we propose a more generalized formulation for coded masks, wherein the elements of $\boldsymbol{d}_l$, $l=1,\ldots,L$, are independently and identically distributed (i.i.d.) realizations of a  random variable $g \in \mathbb{C}$, subject to the following probabilistic constraints: 
\begin{defn}
 \label{def: g}
  Let $g \in \mathbb{C}$ be a complex-valued random variable. For some parameter $\nu \in [1, \infty)$, $g$ satisfies the following moment conditions: (1) $|g| \leq \nu$; (2) $\mathbb{E} g = 0$; (3) $\mathbb{E} |g|^2 = 1$.
   \end{defn}

The first moment condition, $|g|\leq \nu$, is widely recognized as an admissibility criterion in the context of random modulation schemes \cite{ASH20}. If $g$ does not satisfy the second and third moment conditions, these can be met through appropriate shifting and scaling. Notably, the Rademacher-distributed random mask represents a special case of our generalized model, corresponding to the parameter value $\nu = 1$. In practical applications, binary masks are commonly $\{0, 1\}$-valued. To convert these into $\{\pm 1\}$-valued masks, we can employ an additional all-one mask, as described in \cite{Romberg_RandomMask}.

When $g$ is complex-valued, as in the case of phase masks, a pertinent example is given by the following discrete probability distribution: 
\begin{equation}
\label{eqn: proper distribution}
 g=
 \begin{cases} 
 1, & \text{with probability } 1/4;\\
 -1, & \text{with probability } 1/4;\\
  \mathrm{i}, & \text{with probability } 1/4;\\
   -\mathrm{i}, & \text{with probability } 1/4. 
   \end{cases}
    \end{equation}
This distribution, which we shall refer to as the quaternary phase distribution, is symmetric about the origin in the complex plane and uniformly distributed on the unit circle's principal axes \cite{CDP}. For this particular choice of $g$, we again have $\nu = 1$, thereby preserving the fundamental properties of the modulations while extending to the complex domain.

Under the generalized random mask assumption, we consider the scenario of $\ell_2$-norm bounded noise corrupted observations, which are modeled as:
\begin{equation}
\label{noise_time_domain}
\boldsymbol{y}_{l}=\boldsymbol{h}\circledast\left(\boldsymbol{d}_{l}\odot\boldsymbol{x}\right)+\vz_l,\qquad\text{for}\  l=1,\cdots,L.
\end{equation}
Let the noise matrix $\vZ$ be defined as:
\begin{equation}
\label{eqn: z_term}
\vZ: = [\vz_1, \ldots, \vz_L] \in \mathbb{C}^{n \times L}.
\end{equation} 
 We assume that the Frobenius norm of $\vZ$, $\|\vZ\|_F$, which can be equivalently expressed as the $\ell_2$-norm of the vectorized form of $\vZ$, is bounded.

The primary objectives of this paper are to analyze the observation model in (\ref{noise_time_domain}) under general mask assumption described in Definition \ref{def: g} and address the following two fundamental questions:
\begin{enumerate} [(i)]
\item For general signals $\boldsymbol{h}, \boldsymbol{x} \in \mathbb{C}^n$, is it possible to robustly recover $\boldsymbol{h}$ and $\boldsymbol{x}$ in the presence of $\ell_2$-noise perturbation, such that the error bound cannot be further improved?
 \item When the signal $\boldsymbol{x}$ is subject to a sparsity constraint, can the sampling complexity $L$ required for robust recovery be improved under a specific recovery model?
  \end{enumerate}

\subsection{Related Works}
\subsubsection{General Signal Case}
In the noiseless model described by (\ref{eq: original measurement}), a convex programming technique known as "lifting" was proposed, which reformulates the blind deconvolution problem as the estimation of a rank-1 matrix \cite{ARR14, Romberg_RandomMask}. The observations in (\ref{eq: original measurement}) are identical to those in \cite{Romberg_RandomMask}, but without the subsampling applied.  In \cite{Romberg_RandomMask}, Bahmani and Romberg demonstrated that, when $\boldsymbol{d}_l$, $l = 1, \dots, L$, are Rademacher random vectors, the reconstruction of $\widehat{\boldsymbol{h}} \boldsymbol{x}^T$ can be achieved through nuclear-norm minimization, provided the sampling complexity satisfies
 \[
L\gtrsim\mu\log^{2}n\log(n/\mu)\log\log(n+1),
\]  
where 
\begin{equation}
\label{eqn: mu}
\mu=\|\widehat{\boldsymbol{h}}\|_{\infty}^{2}/\|\boldsymbol{h}\|_{2}^{2}
\end{equation} 
represents the coherence parameter of the blurring kernel $\boldsymbol{h}$. Here $\widehat{\boldsymbol{h}}$ is the discrete fourier transform of $\vh$. However, their results were restricted to the noiseless case, and no conclusions were drawn for scenarios involving noise.

The sampling complexity was subsequently improved by Lin and Strohmer \cite{BD_LS}, who investigated the minimization of a least squares problem. In their work, they showed that the optimal solution $\boldsymbol{z}^{\#}$ for the least squares problem satisfies
\begin{equation}\label{eqn: LS_error}
\frac{\|{\boldsymbol{z}^{\#}}-\tau\boldsymbol{z}_0\|_{2}}{\|\tau\boldsymbol{z}_0\|_{2}}\leq\kappa(\mathcal{A}_{\boldsymbol{w}})\eta\left(1+\frac{2}{1-\kappa(\mathcal{A}_{\boldsymbol{w}})\eta}\right),
\end{equation}
provided that $L \gtrsim \log^2 n$, where $\boldsymbol{z}_0 = \begin{bmatrix} \boldsymbol{s} \\ \boldsymbol{x} \end{bmatrix}$ with $\boldsymbol{s}$ being the element-wise inverse of $\widehat{\boldsymbol{h}}$, and $\tau = \frac{c}{\boldsymbol{w}^* \boldsymbol{z}_0}$ for suitably chosen $\boldsymbol{w} \in \mathbb{C}^{2n}$ and $c \in \mathbb{C}$. The noise level $\eta$ is related to the Frobenius norm of $\vZ$ as described in (\ref{eqn: z_term}) in the frequency domain, specifically, $\eta = \|\vF \vZ\|_F = {\sqrt{n}} \|\vZ\|_F,$ where $\vF$ is the $n \times n$ discrete Fourier transform (DFT) matrix, with its $(j,k)$-th element given by $F_{j,k} = \exp\left(-\frac{2\pi \mathrm{i}(j-1)(k-1)}{n}\right).$ Additionally, $\kappa(\mathcal{A}_{\boldsymbol{w}})$  denotes the condition number of the matrix $\mathcal{A}_{\boldsymbol{w}}$, which is explicitly given in \cite{BD_LS}. It follows from this analysis that robust recovery is achievable only when the noise level $\eta$ satisfies the condition $\eta \leq \frac{1}{\kappa(\mathcal{A}_{\boldsymbol{w}})}.$

\subsubsection{Subspace Signal or Sparse Signal Case. }
If $\boldsymbol{x}$ resides within a specified subspace, expressed as $\boldsymbol{x} = \boldsymbol{D}\boldsymbol{z}$, where $\boldsymbol{D}$ is a known $n \times K$ tall orthonormal matrix, a sequence generated by a gradient descent algorithm will converge to the true solution in the noiseless case, provided that the sampling complexity satisfies $L \gtrsim \nu^2(\mu^2\nu^2_{\max}\frac{KL^2}{n} + \nu^2)\log^4n$ \cite{BD_randomSign}. Here, $\nu_{\max}^2 = n \|\boldsymbol{D}\|_\infty^2$ and $\widetilde{\nu}^2 = \frac{n \|\boldsymbol{D} \boldsymbol{z}\|_\infty^2}{\|\boldsymbol{z}\|_2^2}$. However, when $\boldsymbol{D}$ is chosen as the first $K$ columns of the identity matrix $\boldsymbol{I}$, it follows that $\nu_{\max}^2 = n$ and $\widetilde{\nu}^2 \geq 1$. In this case, it becomes untenable to satisfy the condition  $L \gtrsim \widetilde{\nu}^2(\mu^2\nu^2_{\max}\frac{KL^2}{n} + \widetilde{\nu}^2)\log^4n$ even for sufficiently large $L$.

When sparsity is incorporated into the recovery model, several algorithmic advancements have been proposed to address blind deconvolution problem \cite{CE16, SAS, MC99, Ji,XS18}. In the context of calibration, Corollary 3 in \cite{CompressedDeconvolution} shows that when $\boldsymbol{h} \in \mathbb{C}^n$
  is known, the true signal $\boldsymbol{x}$ can be recovered with probability at least $1 - 2 \exp(-Ck)$ for some positive constant $C$, provided that the sampling complexity satisfies $L \gtrsim K \log n$, assuming $\boldsymbol{x}$ is $K$-sparse. However, in our self-calibration problem, $\boldsymbol{h}$ is unknown. Moreover, when $K$ becomes sufficiently large, the required sampling complexity in \cite{CompressedDeconvolution} exceeds the typical $O(\text{polylog}(n))$ bound for the general $\boldsymbol{x}$ case. The results in \cite{Qiao2025} indicate that the proximal alternating linearized minimization (PALM) model can effectively recover both $\boldsymbol{x}$ and $\boldsymbol{h}$ in noisy environments. Convergence results demonstrate that the generated sequence converges to a critical point of the corresponding optimization model. Nevertheless, the sampling complexity required for successful or robust recovery remains unknown. 

In summary, although several studies have tackled the sampling complexity and algorithmic guarantees for sparse recovery, the existing sampling complexities for self-calibration problem  remain suboptimal and have yet to fully address the challenges posed by noisy scenarios.

\subsection{{Our Contributions}}
In this paper, we address the questions posed in Section \ref{sec: problems} when the observations are corrupted by $\ell_2$
 -bounded noise. One of our key contributions is the extension of theoretical analysis to encompass a broader class of random masks. Unlike previous studies that predominantly focused on Rademacher masks, our work introduces a more general mask framework as defined in Definition \ref{def: g}. This generalization allows for a more comprehensive understanding of various mask properties and their effects on signal reconstruction performance.

With respect to Question (i), our analysis centers on the reconstruction of arbitrary signals $\vh$ and $\vx$ within the framework of the constrained least squares model delineated in (\ref{eqn: noise1}). Theorem \ref{thm: stable} establishes that when the sampling complexity adheres to the condition $L \gtrsim \mu C_{\nu} \text{polylog}(n)$, the recovery error is provably bounded above by $\sqrt{n} \|\vZ\|_F$, where $\vZ$ is precisely defined in equation (\ref{eqn: z_term}). The optimality of this result is substantiated by Theorem \ref{thm: lower_bound}, which demonstrates the near-tightness of the error bound established in Theorem \ref{thm: stable}.

Addressing Question (ii), we demonstrate that the introduction of a sparsity constraint leads to a significant reduction in the required sampling complexity. Theorem \ref{thm: h_recovery} delineates a robust reconstruction scheme for the recovery of $\vh$ in the presence of $\ell_2$-bounded noise. Under the assumptions that the blur kernel $\vh$ exhibits compact support and the signal $\vx$ is $k$-sparse and its non-zero elements are sufficiently separated, the sampling complexity $L$ attains the  order of ${O}(\log n)$, see as in Corollary \ref{coro: h_optimal}. Subsequent to the estimation of $\vh$, we employ a LASSO-based model for the recovery of $\vx$, which can be formulated as a constrained optimization problem in the $\ell_1$-regularized least squares framework. The theoretical guarantees for this recovery process are expounded in Theorem \ref{thm: x_estimation}.

Our research extends beyond theoretical analysis to incorporate extensive numerical experiments. These experiments explore the performance of our proposed methods under a variety of mask settings, encompassing both real and complex cases. We conduct extensive simulations to empirically validate the superior performance of the constrained least squares model (\ref{eqn: noise1}) in comparison to the least squares model proposed in previous work for general configurations of $\vx$ and $\vh$. The results of these comparisons are illustrated in Figure \ref{fig: robust_comparison0}. Furthermore, by exploiting the sparsity prior of $\vx$, we demonstrate the enhanced efficacy of the Proximal Alternating Linearized Minimization (PALM) algorithm with specific initializations $\vh_0$ and $\vx_0$, as delineated in Algorithm \ref{alg1}. Our approach exhibits superior performance compared to other state-of-the-art algorithms, as evidenced in our numerical results presented in Figure \ref{fig: comparison2} and Figure \ref{fig: comparison3}. To further validate our approach, we extend our analysis to two-dimensional imaging applications, adopting a setup for blind deconvolution in random mask imaging. The visual results of this extension are shown in Figure \ref{fig: output}, demonstrating the effectiveness of our method in practical scenarios.

\subsection{Notations and Definitions}
Matrices and vectors are denoted by boldface uppercase and lowercase letters, respectively. For a vector $\boldsymbol{x} \in \mathbb{C}^n$, the $p$-norm ($1 \leq p \leq \infty$) is defined as $\|\boldsymbol{x}\|_p = \left( \sum_{j=1}^n |x_j|^p \right)^{1/p}.$ For a matrix $\boldsymbol{X} \in \mathbb{C}^{m \times n}$, we denote the operator norm, nuclear norm  and the Frobenius norm of $\boldsymbol{X}$ as $\|\boldsymbol{X}\|$, $\|\vX\|_{*}$ and $\|\boldsymbol{X}\|_F$, respectively. The transpose and the complex conjugate transpose of $\boldsymbol{X}$ are written as $\boldsymbol{X}^T$ and $\boldsymbol{X}^*$. For any matrices $\boldsymbol{X}, \boldsymbol{Y} \in \mathbb{C}^{m \times n}$, the inner product is defined as $\langle \boldsymbol{X}, \boldsymbol{Y} \rangle := \text{Tr}(\boldsymbol{X}^* \boldsymbol{Y})$.  Additionally, we denote $C_{\alpha}$
  as a constant that depends on the parameter $\alpha$. We use the notation $A \gtrsim B$ to indicate that $A \geq C B$, where $C$ is a positive absolute constant. Similarly, $A \lesssim B$ is defined in the same way. Furthermore, $A = \Theta(B)$ means that there exist positive absolute constants $C_1$ and $C_2$ such that $C_1 B \leq A \leq C_2 B$.

Let $\vF$ denote the $n \times n$ Discrete Fourier Transform (DFT) matrix, where the $(j,k)$-th element is given by $F_{j,k} = \exp\left(-\frac{2\pi \mathrm{i}(j-1)(k-1)}{n}\right).$ For a vector $\boldsymbol{z} = [z_1, \dots, z_n]^T \in \mathbb{C}^n$, the matrix $\text{diag}(\boldsymbol{z}) \in \mathbb{C}^{n \times n}$ is the diagonal matrix with the $k$-th diagonal entry $z_k$. Moreover, the cyclic shift $s_{\tau}(\boldsymbol{z})$
  of $\boldsymbol{z}$ for some $\tau \in \{0, \ldots, n\}$ is denoted as
\begin{equation}\label{eqn: s_tau}
s_{\tau}(\boldsymbol{z}):=[{z}_{l(1-\tau)},\ldots,{z}_{l(n-\tau)}]^{T},
\end{equation}
where 
\[
l(j)=\begin{cases}
j, & \text{if}\ j>0;\\
j+n, & \text{otherwise.}
\end{cases}
\]
We also denote $\vC_{\vz}$ and $\check{\vC}_{\vz}$ respect to vector $\vz$ as follows:
\begin{equation}
\label{eqn: check_C}
 \vC_{\vz}=\begin{bmatrix}z_{1} & z_{n} & \cdots & z_{2}\\
z_{2} & z_{1} & \cdots & z_{3}\\
\vdots & \vdots & \ddots & \vdots\\
z_{n} & z_{n-1} & \cdots & z_{1}
\end{bmatrix}\quad \text{and}\quad \check{\vC}_{\vz}:=\begin{bmatrix}z_1 & z_2 & \cdots & z_n\\
z_2 & z_3 & \cdots & z_1\\
\vdots & \vdots & \ddots & \vdots\\
z_n & z_1 & \cdots & z_{n-1}
\end{bmatrix}.
\end{equation}
\subsection{Orgainizations}
The structure of this paper is as follows:
Section \ref{sec: main_result} presents the principal theoretical contributions of our work. This includes: (i) Stability theorems for general signal reconstruction, establishing near optimal error bounds and sampling complexity requirements. (ii) Refined stability results for sparse signals, incorporating sparsity priors to achieve improved recovery guarantees. (iii) Algorithmic frameworks that bridge theoretical insights with practical implementations for efficient signal recovery. Section \ref{sec: numerical} is dedicated to extensive numerical experiments, providing empirical validation of our theoretical results and demonstrating the practical efficacy of our proposed methods across various scenarios.
Sections \ref{sec: stable} through \ref{sec: noise_upper} contain rigorous mathematical proofs of the main theorems and crucial intermediate results. 

\section{Main Results}\label{sec: main_result}
The theoretical results presented herein consider random masks under the general assumptions outlined in Definition \ref{def: g}. It broadens the scope of analysis compared to previous theoretical works, which primarily focused on Rademacher masks \cite{BD_LS,BD_randomSign, Romberg_RandomMask}. By adopting a more generalized framework, our analysis offers a comprehensive understanding of mask properties and their impact on signal reconstruction.
\subsection{General Signal Case.}
To align the model and results delineated in \cite{Romberg_RandomMask}, we rewrite the observations in (\ref{noise_time_domain}) by introducing a normalization factor, specifically:
\[
{\vy}_l^{\mathrm{normal}}=\frac{1}{\sqrt{L}}\vh\circledast(\vd_l\odot \vx)+\frac{1}{\sqrt{L}}\vz_l,\quad \text{for}\ l=1,\ldots,L.
\]
By applying the discrete Fourier transform to the observations ${\boldsymbol{y}}_l$, for $l = 1, \dots, L$, in (\ref{noise_time_domain}), the observations can be transformed as follows:
\begin{equation}\label{eqn: equivalent model}
\widehat{{\boldsymbol{y}}}_{l}^{\mathrm{normal}}=\frac{1}{\sqrt{L}}\widehat{\boldsymbol{h}}\odot\widehat{\boldsymbol{D}_l\boldsymbol{x}}+\frac{1}{\sqrt{L}}\widehat{\vz_l}=\frac{1}{\sqrt{L}}(\vF\odot(\widehat{\vh}\vx^{T}))\vd_{l}+\frac{1}{\sqrt{L}}\widehat{\vz_l},\qquad\text{for}\ \   l=1,\ldots,L,
\end{equation}
where $\boldsymbol{D}_{l}:=\text{diag}(\boldsymbol{d}_{l})$ and $\widehat{\boldsymbol{z}}:=\boldsymbol{F}\boldsymbol{z}$ for any $\boldsymbol{z}\in \mathbb{C}^n$.
Defining the following:
\begin{equation}
\label{eqn: Y_g}
\widehat{\boldsymbol{Y}} := [\widehat{{\boldsymbol{y}}}_{1}^{\mathrm{normal}}, \dots, \widehat{{\boldsymbol{y}}}_{L}^{\mathrm{normal}}],\quad\widehat{\vZ}:=\frac{1}{\sqrt{L}}[\widehat{\vz}_1,\ldots,\widehat{\vz}_L],\quad  \text{and}\quad \vD_g: = [\vd_1, \dots, \vd_L],
\end{equation}
we can obtain that 
\begin{equation}
\label{eqn: Y_hat}
{\widehat{\boldsymbol{Y}}}=\mathcal{A}({\widehat{\boldsymbol{h}}\boldsymbol{x}^{T}})+\widehat{\vZ},
\end{equation}
where the operator $\mathcal{A}:\mathbb{C}^{n\times n}\rightarrow \mathbb{C}^{n\times L}$ is defined as:
\begin{equation}
\label{eqn: A_operator}
\mathcal{A}(\vX):=\frac{1}{\sqrt{L}}(\vF\odot \vX)\vD_{g},
\end{equation}
for any $\vX\in \mathbb{C}^{n\times n}$. It is the same linear operator as that presented in \cite{Romberg_RandomMask}, excluding the subsampling procedure. However, they did not provide any results for the recovery error bound in the noisy case, nor did they assess whether the error bound is optimal.

First of all, we examine the error bound under the following constrained least squares model:
\begin{equation}
\label{eqn: noise1}
\min_{\boldsymbol{X} \in \mathbb{C}^{n \times n}} \quad \|\widehat{\vY} - \mathcal{A}(\boldsymbol{X})\|_F \quad \text{s.t.} \quad \|\boldsymbol{X}\|_* \leq R.
\end{equation}
Theorem \ref{thm: stable} establishes a robust error bound for signal reconstruction in the presence of $\ell_2$-corrupted noise. This result offers valuable insights into the relationship between the sampling requirements and reconstruction accuracy in noisy scenarios.
\begin{thm}\label{thm: stable}
Consider a sequence of independent diagonal matrices $\boldsymbol{D}_l = \mathrm{diag}(\vd_l)$, $l = 1, \dots, L$, where the diagonal entries are independent copies of a random variable $g \in \mathbb{C}$ with parameter $\nu$ as defined in Definition \ref{def: g}. For any fixed $\widehat{\boldsymbol{h}},\boldsymbol{x}\in \mathbb{C}^n$, assume that $\|\widehat{\boldsymbol{h}}\|_2 = \|\boldsymbol{x}\|_2 = 1$. The observations are corrupted by a noise term $\widehat{\vZ} \in \mathbb{C}^{n \times L}$, so that the observed matrix is given by:
$\widehat{\vY} = \mathcal{A}(\widehat{\boldsymbol{h}} \boldsymbol{x}^T) + \widehat{\vZ},$ as in (\ref{eqn: Y_hat}). Then, with probability at least $1 - C_1 n^{-1} - \exp\left( -\frac{n}{4C^2 \nu^4} \right)$, the solution $\boldsymbol{X}^{\#}$ to the optimization problem (\ref{eqn: noise1}) with $R = \|\widehat{\boldsymbol{h}} \boldsymbol{x}^T\|_*$
 satisfies the error bound:
\begin{equation}
\label{eqn: noise_est}
\|\boldsymbol{X}^{\#} - \widehat{\boldsymbol{h}} \boldsymbol{x}^T\|_F \lesssim \sqrt{n} \|\widehat{\vZ}\|_F,
\end{equation}
 provided that $L \gtrsim C_{\nu} \mu \log^2 n \log\left({n}/{\mu}\right) \log\log n$ and $n \geq \max\left\{16 C^2 \nu^4, 1\right\} L$. Here, $\mu = {\|\widehat{\boldsymbol{h}}\|_\infty^2}/{\|{\boldsymbol{h}}\|_2^2}$, $C$ and $C_1$ are absolute constants, $C_{\nu}$  is a constant depending on $\nu$.\end{thm}
\begin{proof}
The proof is postponed in Section \ref{sec: stable}.
\end{proof}
\begin{remark}
Based on Corollary 1 in \cite{ARR14}, which follows a similar line of reasoning as the latter part of Theorem 1.2 in \cite{phaselift}, we can obtain estimates for $\widehat{\vh}$
  and $\vx$ through the singular value decomposition of $\vX^{\#}$. More concretely, Let $\sigma^{\#} \vu^{\#} (\vv^{\#})^*$ denote the best rank-1 approximation to $\vX^{\#}$, and define $\widehat{\vh}^{\#} = \sqrt{\sigma^{\#}} \vu^{\#}$ and $\vx^{\#} = \sqrt{\sigma^{\#}} \vv^{\#}$. Then, we have the following error bounds:
\[
\|\widehat{\vh}^{\#}-\alpha\widehat{\vh}\|_2\lesssim \min(\sqrt{n}\|\widehat{\vZ}\|_F,\|\widehat{\vh}\|_2),\qquad\text{and}\qquad \|{\vx}^{\#}-\alpha^{-1}{\vx}\|_2\lesssim \min(\sqrt{n}\|\widehat{\vZ}\|_F,\|{\vx}\|_2),
\]
for some non-zero scalar $\alpha$. 
\end{remark}
\begin{remark}
The constrained least squares model in (\ref{eqn: noise1}) is closely related to problems in areas like phase retrieval \cite{HuangXu, XiaXu}. In particular, consider the following nuclear norm minimization model:
\begin{equation}\label{eqn: nuclear_norm min}
\min_{\boldsymbol{X}\in \mathbb{C}^{n\times n}}\ \|\boldsymbol{X}\|_*\qquad \mathrm{s.t.}\ \|\mathcal{A}(\boldsymbol{X})-{\widehat{\vY}}\|_F\leq \|\widehat{\vZ}\|_F.
\end{equation}
In this context, the constrained least squares model  (\ref{eqn: noise1}) can be viewed as the dual model of the nuclear norm minimization model (\ref{eqn: nuclear_norm min}).

Efficient numerical methods are available for solving the nuclear norm minimization problem. For small-scale problems, semidefinite programming (SDP) solvers can be used. For larger-scale cases, matrix factorization-based algorithms are a more scalable approach to solving the semidefinite program \cite{BM03, SDP_converge}. Building on these factorization techniques, the scaled gradient method has been shown to achieve rapid and reliable convergence in matrix recovery problems \cite{Tong21}. These methods thus provide effective tools for solving (\ref{eqn: noise1}).
\end{remark}

\begin{remark}
The assumption that $\|\widehat{\boldsymbol{h}}\|_{2} = \|\boldsymbol{x}\|_{2} = 1$ is made for the sake of convenience in the proof, as noted in \cite{Romberg_RandomMask}. It is important to clarify that this is not a fundamental assumption of our results. In \cite{Romberg_RandomMask}, the authors do not address the recovery error in the presence of noise, whereas our analysis extends to include this consideration. Furthermore, our error bound does not impose additional conditions on the noise level, contrasting with the findings in \cite{BD_LS}. 

Additionally, sampling complexity can be improved using more refined techniques, as  in \cite{Gross}, which discusses reducing the number of measurements in the phase retrieval problem through coded diffraction patterns. However, our focus in Theorem \ref{eqn: noise_est} is primarily on presenting the noise bound. The exploration of improvements to sampling complexity is left for the interested reader to pursue further.
\end{remark}

The following theorem establishes a lower bound on the reconstruction error, demonstrating the optimality of our previously derived error bound in Theorem \ref{thm: stable}. This result provides a complementary perspective by showing that, under certain conditions, there exists a noise matrix that leads to a reconstruction error of a similar order  as the upper bound. This theoretical finding underscores the tightness of our error analysis and reinforces the robustness of model (\ref{eqn: noise1}) in handling noisy measurements.

\begin{thm}\label{thm: lower_bound}
Consider a sequence of independent diagonal matrices $\boldsymbol{D}_l = \mathrm{diag}(\vd_l)$, $l = 1, \dots, L$, where the diagonal entries are independent copies of a random variable $g \in \mathbb{C}$ with parameter $\nu$ as defined in Definition \ref{def: g}.  Let $\widehat{\boldsymbol{h}},\boldsymbol{x}\in \mathbb{C}^n$ be fixed vectors such that  $\|\widehat{\boldsymbol{h}}\|_2=1$ and  $\|\boldsymbol{x}\|_2 = 1$ excluding a set of measure zero. The observations are corrupted by a noise term $\widehat{\vZ} \in \mathbb{C}^{n \times L}$, so that the observed matrix is given by:
$\widehat{\vY} = \mathcal{A}(\widehat{\boldsymbol{h}} \boldsymbol{x}^T) + \widehat{\vZ},$ as in (\ref{eqn: Y_hat}).
Assume that $L = \Theta \left( C_{\nu} \mu \log^2{n} \log\left(\frac{n}{\mu}\right) \log\log{n} \right)$ and $n \geq 32 \nu^6 L \log n$. Then, with probability at least $1 - \frac{4}{n}$, there exists a noise matrix $\widehat{\boldsymbol{Z}}$ such that the corresponding optimal solution $\boldsymbol{X}^{\#}$
  to the optimization problem (\ref{eqn: noise1}) with $R = \|\widehat{\boldsymbol{h}} \boldsymbol{x}^T\|_*$ satisfies the following lower bound: 
  \begin{equation}
 \label{eqn: error_lower}
 \|\boldsymbol{X}^{\#}-\widehat{\boldsymbol{h}}\boldsymbol{x}^{T}\|_F\gtrsim {\frac{\sqrt{n}}{\widetilde{C}_{\nu}{\mu}^{1/2}\log^3 n}}\|\widehat{\boldsymbol{\vZ}}\|_F,
 \end{equation}
where $C_{\nu}$ and $\widetilde{C}_{\nu}$ are positive constants depending on the parameter $\nu$, and $\mu = {\|\widehat{\boldsymbol{h}}\|_\infty^2}/{\|{\boldsymbol{h}}\|_2^2}$.
\end{thm}
\begin{proof}
The proof is postponed in Section \ref{sec: Thm2}.
\end{proof}
\begin{remark}
If $\widehat{\vh}$ is drawn from the uniform distribution on the complex unit sphere in $\mathbb{C}^n$, then the parameter $\mu$ can be upper bounded by $\log^2 n$ with high probability. In this scenario, although the error bound presented in (\ref{eqn: noise_est}) appear relatively large due to its dependence on the signal dimension $n$, Theorem \ref{thm: lower_bound} rigorously establishes that this bound is, in fact, tight up to a logarithmic factor.
\end{remark}
\subsection{Sparse Signal Case} 

The following results address the scenario where the signal vector $\boldsymbol{x}$ exhibits sparsity. In this context, we develop independent schemes for the robust recovery of both $\boldsymbol{h}$ and $\boldsymbol{x}$, leveraging the sparse structure of $\boldsymbol{x}$ to enhance reconstruction performance. Theorem \ref{thm: h_recovery} focuses specifically on the recovery of $\boldsymbol{h}$, providing recovery guarantees in the presence of noise. 

\begin{thm}\label{thm: h_recovery}
Consider a sequence of independent diagonal matrices $\boldsymbol{D}_l = \mathrm{diag}(\vd_l)$, $l = 1, \dots, L$, where the diagonal entries are independent copies of a random variable $g \in \mathbb{C}$ with parameter $\nu$ as defined in Definition \ref{def: g}.  For any fixed nonzero ${\boldsymbol{h}}$ and $\boldsymbol{x}$, assume that $\|\boldsymbol{x}\|_{2}=1$ and $\boldsymbol{x}$ is  $K$-sparse. The observations are given by ${\boldsymbol{y}}_l=\boldsymbol{h}\circledast\left(\boldsymbol{d}_{l}\odot\boldsymbol{x}\right)+\vz_l$, $l=1,\ldots, L$, as in (\ref{noise_time_domain}), where the noise matrix $\boldsymbol{Z}$ in (\ref{eqn: z_term}) satisfies $\|\vZ\|_F\leq C\|\vx\|_2\|\vh\|_2$ for some absolute constant $C > 0$. Define the matrix $\boldsymbol{H}$ as
\begin{equation}\label{eqn: H}
\boldsymbol{H}=\frac{1}{{{L}}}\sum_{l=1}^L\overline{\boldsymbol{D}_l}\check{\vC}_{{\boldsymbol{y}}_l},
\end{equation}
where $\check{\vC}_{{\boldsymbol{y}}_l}$ is defined as in (\ref{eqn: check_C}).
 Let $\boldsymbol{h}_j \in \mathbb{C}^n$ denote the $j$-th row of $\boldsymbol{H}$ for $j = 1, \dots, n$. Take $j^{\#}: = \operatorname{argmax}_{j} \|\boldsymbol{h}_j\|_2$. Then, with probability at least $1 - 1/n$, we have
\[
\mathrm{dist}\left(\frac{\vh_{j^{\#}}}{\|\vh_{j^{\#}}\|_2},\frac{\boldsymbol{h}}{\|\boldsymbol{h}\|_2}\right)\leq \widetilde{\epsilon}+ C\widetilde{\epsilon}_1,
\]
provided that 
\begin{equation}
\label{eqn: lower_L}
L\gtrsim \frac{\sqrt{1+\mu_h K}\nu^2}{\|\boldsymbol{x}\|_\infty^2\min\{\widetilde{\epsilon}^2,\widetilde{\epsilon}_1^2\}}\log n,
\end{equation}
where $\widetilde{\epsilon}$ and $\widetilde{\epsilon}_1$ are positive constants satisfying $\widetilde{\epsilon}<4$ and $C\widetilde{\epsilon}_1<2$. Here, the distance function $\mathrm{dist}(\boldsymbol{x}, \boldsymbol{y})$ for $\boldsymbol{x}, \boldsymbol{y} \in \mathbb{C}^{n}$ and the mutual coherence parameter $\mu_h$ are defined as follows:
 \begin{equation}
 \label{eqn: dist_mu}
 \mathrm{dist}(\boldsymbol{x},\boldsymbol{y})=\min_{\theta\in [0,2\pi)}\|\boldsymbol{x}-e^{\mathrm{i}\theta}\boldsymbol{y}\|_2\quad \text{and}\quad \mu_{h}=\max_{i,j\in \mathrm{supp}(\boldsymbol{x}), i\neq j}\frac{|\langle s_{i}(\boldsymbol{h}),s_{j}(\boldsymbol{h})\rangle|}{\|\boldsymbol{h}\|_2^2}
 \end{equation}
 with $s_\tau(\cdot)$ as in (\ref{eqn: s_tau}) for any $\tau\in\{1,\dots,n\}$. 
\end{thm}
\begin{proof}
The proof is postponed in Section \ref{sec: Thm3}.
\end{proof}
\begin{remark}
Let us define the linear operator $\widetilde{\mathcal{A}}_l(\cdot): \mathbb{C}^{n \times n} \rightarrow \mathbb{C}^n$ as follows:
\[
\widetilde{\mathcal{A}}_l(\boldsymbol{x}\boldsymbol{h}^{T}) := \boldsymbol{h} \circledast (\vd_l \odot \boldsymbol{x}), \qquad l = 1, \ldots, L.
\]
 Take $ \widetilde{\mathcal{A}}_l^*$ as the dual operator for $ \widetilde{\mathcal{A}}_l$. The motivation for constructing $\vH$ is given by the expression $\vH = \frac{1}{L} \sum_{l=1}^L \widetilde{\mathcal{A}}_l^*(\vy_l)$ and  $\mathbb{E}\vH = \vx \vh^T$, particularly when $\|\vZ\|_F = 0$. However, rather than employing the spectral method as outlined in \cite{BD_randomSign,Candes15}, we adopt a more straightforward approach to retrieve $\vh$.
\end{remark}

\begin{remark}
Due to the inherent scaling ambiguity in self-calibration, that is, for any non-zero $\alpha$, we have $\vy_l = (\alpha \vh) \circledast \left( \vd_l \odot \left( \frac{1}{\alpha} \vx \right) \right) + \vz_l$ for $l = 1, \ldots, L$, it is reasonable to impose the constraint that the $\ell_2$ norm of $\vx$ is equal to 1 to remove scaling ambiguity, while leaving the norm of $\vh$ unconstrained. Moreover, the condition on the noise level, such that $\|\vZ\|_F$ is bounded by $\|\vx\|_2 \|\vh\|_2$, implies that the noise should be regarded as adversarial noise that aligns with the signals $\vx$ and $\vh$. In the noiseless scenario, employing analogous technical tools allows us to derive the following inequality:
\[
\mathrm{dist}\left(\frac{\vh_{j^{\#}}}{\|\vh_{j^{\#}}\|_2},\frac{\boldsymbol{h}}{\|\boldsymbol{h}\|_2}\right)\leq \widetilde{\epsilon},
\]
provided that 
\[
L\gtrsim \frac{\sqrt{1+\mu_h K}\nu^2}{\|\boldsymbol{x}\|_\infty^2\widetilde{\epsilon}^2}\log n.
\]
Thus, we observe that the parameters $\widetilde{\epsilon}$ and $\widetilde{\epsilon}_1$ are utilized to tune the error bounds for exact recovery and noise corruption, respectively. In the presence of noise, when $C < 2$, setting $\widetilde{\epsilon} = C$ and $\widetilde{\epsilon}_1 = 1$ allows us to simplify the error bound to: 
\[
\mathrm{dist}\left(\frac{\vh_{j^{\#}}}{\|\vh_{j^{\#}}\|_2},\frac{\boldsymbol{h}}{\|\boldsymbol{h}\|_2}\right)\leq 2C,
\] given that $L\gtrsim \frac{\sqrt{1+\mu_h K}\nu^2}{\|\boldsymbol{x}\|_\infty^2}\log n$. 
\end{remark}

When $\vh$ is short-supported and the elements in $\vx$ is separated sufficiently,  we can directly get that $\mu_h=0$ in (\ref{eqn: lower_L}), and the sampling complexity $L$ greatly reduced to $L\gtrsim \log n$. More details can be seen in Corollary \ref{coro: h_optimal}. The investigation of improved sampling complexity for general sparse signal $\vx$, without the short support constraint of $\vh$, is left as a topic for future research.

\begin{coro}\label{coro: h_optimal}
Assume that $\|\vx\|_\infty\geq c_0$  for some absolute positive constant $c_0$ and any two non-zero components of $\boldsymbol{x}$ are at least $m$ entries apart. Furthermore, consider $\vh\in \mathbb{C}^{n}$ with $\mathrm{supp}(\boldsymbol{h}) \subseteq \{1, \ldots, m\}$. Under these conditions, the coherence parameter $\mu_h=0$. Subsequently, let $\vH$ be defined as in (\ref{eqn: H}) of Theorem \ref{thm: h_recovery}. Denote by $\boldsymbol{h}_j \in \mathbb{C}^n$  the $j$-th row of $\boldsymbol{H}$ for $j = 1, \dots, n$, and let $j^{\#}: = \operatorname{argmax}_{j} \|\boldsymbol{h}_j\|_2$. With probability at least $1 - 1/n$, we have\[
\mathrm{dist}\left(\frac{\vh_{j^{\#}}}{\|\vh_{j^{\#}}\|_2},\frac{\boldsymbol{h}}{\|\boldsymbol{h}\|_2}\right)\leq \widetilde{\epsilon}+ C\widetilde{\epsilon}_1,
\]
provided that 
\begin{equation}\label{eqn: new_L}
L\gtrsim \frac{\nu^2}{\min\{\widetilde{\epsilon}^2,\widetilde{\epsilon}_1^2\}}\log n,
\end{equation}
where $\widetilde{\epsilon}$ and $\widetilde{\epsilon}_1$ are the same positive constants as in Theorem \ref{thm: h_recovery}.  \end{coro}

Having established a robust method for estimating $\boldsymbol{h}$, we now turn our attention to the estimation of the sparse signal $\boldsymbol{x}$.  If there exists fixed vector ${\vh_0} \in \mathbb{S}^{n-1}$ such that  $\mathrm{dist}(\vh_0, \boldsymbol{h} / \|\boldsymbol{h}\|_2) < \epsilon$, we can employ the LASSO model to reconstruct $\boldsymbol{x}$. This approach is formulated as  below:
\begin{equation}\label{eqn: model01}
\min_{\widetilde{\boldsymbol{x}}\in \mathbb{C}^n}\ \frac{1}{2L}\sum_{l=1}^{L}\|\vh_0\circledast(\vd_l\odot \widetilde{\vx})-\vy_l\|_2^2+\lambda\|\widetilde{\boldsymbol{x}}\|_1,
\end{equation}
which  balances the fidelity to observed data with the promotion of sparsity in the recovered signal. Theorem \ref{thm: x_estimation} provides rigorous performance guarantees for the accuracy of the reconstructed sparse signal under the model  (\ref{eqn: model01}).

\begin{thm}\label{thm: x_estimation}
Consider a sequence of independent diagonal matrices $\boldsymbol{D}_l = \mathrm{diag}(\vd_l)$, $l = 1, \dots, L$, where the diagonal entries are independent copies of a random variable $g \in \mathbb{C}$ with parameter $\nu$ as defined in Definition \ref{def: g}.  For any fixed nonzero ${\boldsymbol{h}}$ and $\boldsymbol{x}$, assume that $\|\boldsymbol{x}\|_{2}=1$ and $\boldsymbol{x}$ is  $K$-sparse. The observations are given by ${\boldsymbol{y}}_l=\boldsymbol{h}\circledast\left(\boldsymbol{d}_{l}\odot\boldsymbol{x}\right)+\vz_l$, $l=1,\ldots, L$, as in (\ref{noise_time_domain}),  where the noise matrix $\boldsymbol{Z}$ in (\ref{eqn: z_term})  satisfies the bound $\|\vZ\|_F\leq C\|\vx\|_2\|\vh\|_2$ for some absolute constant $C > 0$.  Let ${\vh_0} \in \mathbb{S}^{n-1}$ be a fixed vector satisfying $\mathrm{dist}(\vh_0, \boldsymbol{h} / \|\boldsymbol{h}\|_2) < \epsilon$. Let $\theta_0$ denote the optimal phase satisfying $\mathrm{dist}\left(\vh_{0},\boldsymbol{h}/\|\boldsymbol{h}\|_2\right)=\left\|e^{\mathrm{i}\theta_{0}}\vh_{0}-\boldsymbol{h}/\|\boldsymbol{h}\|_2\right\|_2.$ Then, the solution $\boldsymbol{x}^{\#}$  to the LASSO problem (\ref{eqn: model01}) satisfies the error bound:
\begin{equation}
\label{eqn: lasso_bound}
\left\|{\boldsymbol{x}}^{\#}-e^{\mathrm{i} \theta_0} \|\vh\|_2 \vx\right\|_2\leq \min_{1\leq k\leq K}\left(C_1'\sqrt{k}\lambda+C_2'\frac{\|\boldsymbol{h}\|_2\|\boldsymbol{x}-(\boldsymbol{x})_{[k]}\|_1}{\sqrt{k}}\right),
\end{equation}
with probability at least $1-\frac{2}{n}$, provided that $\lambda \geq 2(2\epsilon+C)\|\vh\|_2\|\vx\|_2$ and $L\gtrsim\mu_0\nu^2\log n$. Here 
$\mu_0=\|\widehat{\boldsymbol{h}_0}\|_{\infty}^{2}/\|{\boldsymbol{h}}_0\|_{2}^{2}={\|\widehat{\boldsymbol{{h}}_0}\|_{\infty}^2}$ denotes the coherence parameter of $\vh_0$, and $C_1'$
  and $C_2'$ are absolute positive constants. Furthermore, $(\boldsymbol{x})_{[k]}$ represents the best $k$-sparse approximation of $\boldsymbol{x}$, formed by retaining the $k$ largest entries of $\boldsymbol{x}$ in magnitude and setting the remaining entries to zero.
\end{thm}
\begin{proof}
The proof is postponed in Section \ref{sec: Thm4}.
\end{proof}
\begin{remark}
Based on (\ref{eqn: lasso_bound}),  it is straightforward to observe that $\left\|{\boldsymbol{x}}^{\#}-e^{\mathrm{i} \theta_0} \|\vh\|_2 \vx\right\|_2$ can be upper bounded by 
\[
\left\|{\boldsymbol{x}}^{\#}-e^{\mathrm{i} \theta_0} \|\vh\|_2 \vx\right\|_2 \leq C_1' \sqrt{K} \lambda,
\] since $\vx$ is $K$-sparse. Moreover, the sampling complexity of $O(\log n)$ is substantially more efficient than the $O(K \log n)$ bound established in \cite[Corollary 3]{CompressedDeconvolution} for the recovery of $\vx$.
\end{remark}
Leveraging the results from Theorem \ref{thm: h_recovery}, we can select an appropriate unit-norm vector $\boldsymbol{h}_0$ that closely approximates the normalized true channel $\boldsymbol{h}/\|\boldsymbol{h}\|_2$ such that $\mathrm{dist}\left(\vh_0, \frac{\boldsymbol{h}}{\|\boldsymbol{h}\|_2}\right) < \epsilon,$ where $\epsilon: = \widetilde{\epsilon} + C \widetilde{\epsilon}_1$. However, Theorem \ref{thm: x_estimation} establishes that $\boldsymbol{h}_0$
  must be independent of the measurements.  Consequently, to ensure independent estimation of $\boldsymbol{x}$, we partition the index set $\{1, \ldots, L\}$ into two distinct subsets: $L_1:=\{1,\cdots,\lfloor L/2\rfloor\}$ and $L_2:=\{\lfloor L/2\rfloor +1,\cdots, L\}$. This enables us to derive an explicit construction of  $\vh_0$, which subsequently facilitates robust reconstruction of the sparse signal $\vx$. A detailed exposition of these results is presented in Corollary \ref{coro: x_new}.
  \begin{coro}\label{coro: x_new}
  For any fixed nonzero ${\boldsymbol{h}}$ and $\boldsymbol{x}$, assume that $\|\boldsymbol{x}\|_{2}=1$ and $\boldsymbol{x}$ is  $K$-sparse.  Take the mask matrices $\vD_l$ and the measurements $\vy_l$, $l=1,\ldots,L(L\geq 2)$ with parameters $C$ and $\nu$  the same as those in Theorem \ref{thm: x_estimation}.  Define $\widetilde{\vH}$ as 
\[
\widetilde{\vH}:=\frac{1}{\lfloor L/2\rfloor}\sum_{l\in L_1}\overline{\boldsymbol{D}_l}\check{\vC}_{{\boldsymbol{y}}_l}.
\]
Let $\widetilde{\vh}_{j}$ as the $j$-the row element of $\widetilde{\vH}$, and denote $j_0 := \mathrm{argmax}_j \|\widetilde{\vh}_{j}\|_2$. Take
\begin{equation}
\label{eqn: h_new}
\vh_0:=\widetilde{\vh}_{j_0}/\|\widetilde{\vh}_{j_0}\|_2,
\end{equation}
and 
\begin{equation}\label{eqn: LASSO_new}
\vx_0:=\mathrm{argmin}_{\widetilde{\boldsymbol{x}}\in \mathbb{C}^n}\ \frac{1}{2(L-\lfloor L/2\rfloor)}\sum_{l\in L_2}\|\vh_0\circledast(\vd_l\odot \widetilde{\vx})-\vy_l\|_2^2+\lambda\|\widetilde{\boldsymbol{x}}\|_1.
\end{equation}
Take $\theta_0$ satisfy $\mathrm{dist}\left(\vh_{0},\boldsymbol{h}/\|\boldsymbol{h}\|_2\right)=\left\|e^{\mathrm{i}\theta_{0}}\vh_{0}-\boldsymbol{h}/\|\boldsymbol{h}\|_2\right\|_2$. Then
\begin{equation}
\label{eqn: lasso_bound}
\left\|{\boldsymbol{x}}_0-e^{\mathrm{i} \theta_0} \|\vh\|_2 \vx\right\|_2\leq \min_{1\leq k\leq K}\left(C_1'\sqrt{k}\lambda+C_2'\frac{\|\boldsymbol{h}\|_2\|\boldsymbol{x}-(\boldsymbol{x})_{[k]}\|_1}{\sqrt{k}}\right),
\end{equation}
with probability at least $1-\frac{2}{n}$, provided that $\lambda \geq 2(2\epsilon+C)\|\vh\|_2\|\vx\|_2$ and $L\gtrsim\mu_0\nu^2\log n$. Here 
$\mu_0=\|\widehat{\boldsymbol{h}_0}\|_{\infty}^{2}/\|{\boldsymbol{h}}_0\|_{2}^{2}={\|\widehat{\boldsymbol{{h}}_0}\|_{\infty}^2}$ and  $\mathrm{dist}(\vh_0, \boldsymbol{h} / \|\boldsymbol{h}\|_2) < \epsilon$.  \end{coro}

\begin{remark}
\label{eqn: key_remark}
 By taking $\boldsymbol{h}_0$ from (\ref{eqn: h_new}) and $\boldsymbol{x}_0$  from (\ref{eqn: LASSO_new}), we obtain more reliable approximations of $\boldsymbol{h}$ and $\boldsymbol{x}$ compared to those in Theorem \ref{thm: stable} for small sparsity levels. More concretely, let $\|\widehat{\vh}\|_2 = \sqrt{n}\|\vh\|_2 = 1$ and $\|\vx\|_2 = 1$. Based on Theorem \ref{thm: h_recovery} and Theorem \ref{thm: x_estimation}, let $\theta_0$ be such that $\mathrm{dist}\left(\vh_0, \boldsymbol{h}/\|\boldsymbol{h}\|_2\right) = \left\|e^{\mathrm{i}\theta_0} \vh_0 - \boldsymbol{h}/\|\boldsymbol{h}\|_2\right\|_2$. We then have:
\begin{equation}
\label{eqn: compare1}
\begin{aligned}
\|\widehat{\vh_0}\vx_0^{T}-\widehat{\vh}\vx^{T}\|_F\leq &\sqrt{n}\|\vh_0\vx_0^{T}-\vh\vx^{T}\|_F\\
=&\sqrt{n}\left\|\vh_0\vx_0^{T}-e^{\mathrm{i}\theta_0}\vh_0(\|\vh\|_2\vx^{T})+e^{\mathrm{i}\theta_0}\vh_0(\|\vh\|_2\vx^{T})-\frac{\vh}{\|\vh\|_2}(\|\vh\|_2\vx^{T})\right\|_F\\
\leq & \sqrt{n}\left\|\vx_0-e^{\mathrm{i}\theta_0}\|\vh\|_2\vx\right\|_2\cdot \|\vh_0\|_2+\|\vh\|_2\cdot \|\vx\|_2\cdot \left\|e^{\mathrm{i}\theta_0}\vh_0-\vh/\|\vh\|_2\right\|_2\\
\lesssim & \sqrt{n}\cdot \sqrt{K}\cdot {\lambda}+\frac{1}{\sqrt{n}}\cdot \epsilon
\lesssim \sqrt{K}C
\end{aligned}
\end{equation}
where $\lambda := 2(2\epsilon + C)\|\vh\|_2\|\vx\|_2 = 2(2\epsilon + C)/\sqrt{n}$ and $\epsilon := \widetilde{\epsilon} + C \widetilde{\epsilon}_1$ for some positive absolute constants $\widetilde{\epsilon}$  and $\widetilde{\epsilon}_1$, and $C$ is the noise level satisfying $\|\vZ\|_F \leq C\|\vx\|_2\|\vh\|_2 = C/\sqrt{n}$, with $\vZ$ defined in (\ref{eqn: z_term}). Since $\widehat{\vZ}$, defined in (\ref{eqn: Y_g}), satisfies
\[
\|\widehat{\vZ}\|_F=\|\vF\vZ\|_F/\sqrt{L}=\sqrt{n}\|\vZ\|_F/\sqrt{L}\leq C/\sqrt{L},
\] then error bound in Theorem \ref{thm: h_recovery} directly becomes
\begin{equation}
\label{eqn: compare2}
\|\vX^{\#}-\widehat{\vh}\vx^{T}\|_F\lesssim \sqrt{n}\|\widehat{\vZ}\|_F\leq \frac{\sqrt{n}}{\sqrt{L}}C.
\end{equation}
When the sparsity level $K$ satisfies $K\lesssim n/L$ (Since the order of $L$ is $\mathrm{polylog}(n)$, this can be simplified into $K\mathrm{polylog}(n)\lesssim n$), then the error bound in (\ref{eqn: compare1}) outperforms the error bound in (\ref{eqn: compare2}). 
\end{remark}
\subsection{Refined Algorithm}
 \begin{algorithm}[t]
  \caption{Proximal Alternating Linearized Minimization Algorithm by Specific Initialization}
  \label{alg1}
  \begin{algorithmic}[1]
    \STATE \textbf{Input}: Mask sequences $\{\vd_l\}_{l=1}^L$ and corresponding observations $\{{\boldsymbol{y}}_l\}_{l=1}^L$ as defined in equation (\ref{noise_time_domain}), along with the regularization parameter $\lambda$.
     \STATE \textbf{Initialization}: Obtain  $\boldsymbol{h}_0$ and $\boldsymbol{x}_0$ using  (\ref{eqn: h_new}) and the LASSO model  in (\ref{eqn: LASSO_new}), respectively.
     \STATE 
                        {\textbf{for}} $k=0$ {\textbf{to}} MAXiter:\\
                        Solve $\boldsymbol{h}_{k+1}=\mathrm{argmin}_{\boldsymbol{h}} \langle \vh-\vh_{k},\triangledown_{\vd} F(\vh_{k},\vx_{k})\rangle+\frac{L_k}{2}\|\vd-\vd_{k}\|_2^2$.\\
      Solve $\boldsymbol{x}_{k+1}=\mathrm{argmin}_{\vx} F(\vh_{k+1},\vx)+\lambda\|\vx\|_1$.\\
      \STATE  \textbf{Output}: $\vh^{\#}=\vh_{\text{MAXiter}}$ and $\vx^{\#}=\vx_{\text{MAXiter}}$. 
  \end{algorithmic}
\end{algorithm} 

To further enhance our estimations, we employ the Proximal Alternating Linearized Minimization (PALM) algorithm, detailed in Algorithm \ref{alg1}. This method utilizes the following key functions:
\[
F(\boldsymbol{h},\boldsymbol{x})=\frac{1}{2L}\sum_{l=1}^{L}\|\boldsymbol{h}\circledast(\boldsymbol{d}_l\odot\boldsymbol{x})-\boldsymbol{y}_l\|_2^2\quad \text{and}\quad L_k=\frac{1}{L}\Big\|\sum_{l=1}^{L}\boldsymbol{C}^*_{\boldsymbol{d}_l\odot\boldsymbol{x}_k}\boldsymbol{C}_{\boldsymbol{d}_l\odot\boldsymbol{x}_k}\Big\|,
\]
with $\boldsymbol{C}_{\boldsymbol{d}_l\odot\boldsymbol{x}_k}$ defined as in (\ref{eqn: check_C}).
While this algorithm bears similarities to the approach described in \cite{Qiao2025}, a crucial distinction lies in our focus on directly recovering $\boldsymbol{h}$, rather than $\widehat{\boldsymbol{h}}$ as outlined in \cite{BD_LS}. Notably, the sequence generated by PALM in Algorithm \ref{alg1} converges to a critical point of the following model:
\[
\min_{\boldsymbol{h},\boldsymbol{x}}\frac{1}{2L}\sum_{l=1}^{L}\|\boldsymbol{h}\circledast(\boldsymbol{d}_l\odot\boldsymbol{x})-\boldsymbol{y}_l\|_2^2+\lambda\|\boldsymbol{x}\|_1.
\]
This convergence property holds irrespective of the specific initialization choices for $\boldsymbol{h}_{0}$ and $\boldsymbol{x}_{0}$. However, it's important to note that the number of samples required for recovery is significantly influenced by the selections of $\boldsymbol{h}_{0}$ and $\boldsymbol{x}_{0}$, as elaborated in Theorem \ref{thm: h_recovery} and Theorem \ref{thm: x_estimation}.

Furthermore, our numerical simulations provide compelling evidence that employing constructions of $\boldsymbol{h}_0$ and $\boldsymbol{x}_{0}$ as presented in (\ref{eqn: h_new}) and the model (\ref{eqn: LASSO_new}) leads to substantial improvements. Specifically, we observe enhanced successful recovery rates in the noiseless case and tighter error bounds in the presence of noise. These findings underscore the critical role of appropriate initialization in optimizing algorithm performance. For a comprehensive analysis of these results, we direct the readers to Section \ref{sec: numerical}.

\section{Numerical Experiments}
\label{sec: numerical}
Here we demonstrate the superior performance of the constrained least squares model in (\ref{eqn: noise1}) (abbreviated as \textbf{Constrained LS}) compared to the least squares model in \cite{BD_LS} (abbreviated as \textbf{LS}) for the general case of $\vx$ and $\vh$. Additionally, leveraging the sparsity prior of $\vx$,  we show the enhanced performance of the PALM algorithm with initializations $\vh_0$ and $\vx_0$ as specified in (\ref{eqn: h_new}) and the model (\ref{eqn: LASSO_new}), respectively, relative to other state-of-the-art algorithms in both noiseless and noisy scenarios.
\subsection{Experimental Setup}
 In our numerical results, we consider both real and complex-valued cases.  The non-zero elements of $\vh$ and $\vx$ are independently drawn from standard real Gaussian distribution $\mathcal{N}(0,1)$ and standard complex Gaussian distribution $\frac{1}{\sqrt{2}}\mathcal{N}(0,1)+\frac{\mathrm{i}}{\sqrt{2}}\mathcal{N}(0,1)$ in real and complex-valued cases, respectively. In the real-valued case, we apply Rademacher masks as  in \cite{Romberg_RandomMask, BD_LS}. In the complex-valued case, coded masks with elements independently drawn from (\ref{eqn: proper distribution}) are used. The dimensions of the signal $\vx$ and $\vh$ are set to $n=50$. 

For each algorithm in the experiments, $20$ independent trials are conducted.  Besides, we denote $\vh^{\#}$ and $\vx^{\#}$ as the estimates of $\vh$ and $\vx$, respectively, and define the corresponding relative mean square error (RMSE) as:
\[
\text{RMSE} := \frac{\|\vh^{\#} \vx^{\#} - \vh \vx^{T}\|_F}{\|\vh \vx^{T}\|_F},
\]
due to the fundamental scaling ambiguity of the bilinear problem. 
In the noiseless case, we define a recovery as successful if the RMSE is less than $10^{-3}$. In the noisy case, white Gaussian noise is introduced using the MATLAB function \textbf{awgn}. To assess the quality of the reconstruction under noisy conditions, we compute the signal-to-noise ratio (SNR) of the reconstruction in decibels, given by $-20\log_{10}(\text{RMSE})$. A higher SNR of the reconstruction indicates a better error bound, reflecting more accurate reconstruction.

\subsection{General Signal Case}
In Figure \ref{fig: robust_comparison0}, we set $L = 10$ for both the real and complex cases. The signal-to-noise ratio (SNR) is varied from {$10$dB to $50$dB}, with higher SNR corresponding to lower noise levels. From this figure, we observe that the least squares model in \cite{BD_LS} performs comparably or {even better} at lower noise levels. However, the constrained least squares model in (\ref{eqn: noise1}) outperforms the former in scenarios with lower SNR, or equivalently, higher noise levels. This observation is consistent with the theoretical result, which indicates that the least squares model performs well only under low noise conditions, whereas the constrained least squares model is not dependent on the noise level and exhibits a linear error bound.
\begin{figure}[htbp]
	\centering
	\subfloat[]{{\includegraphics[width=0.45\textwidth]{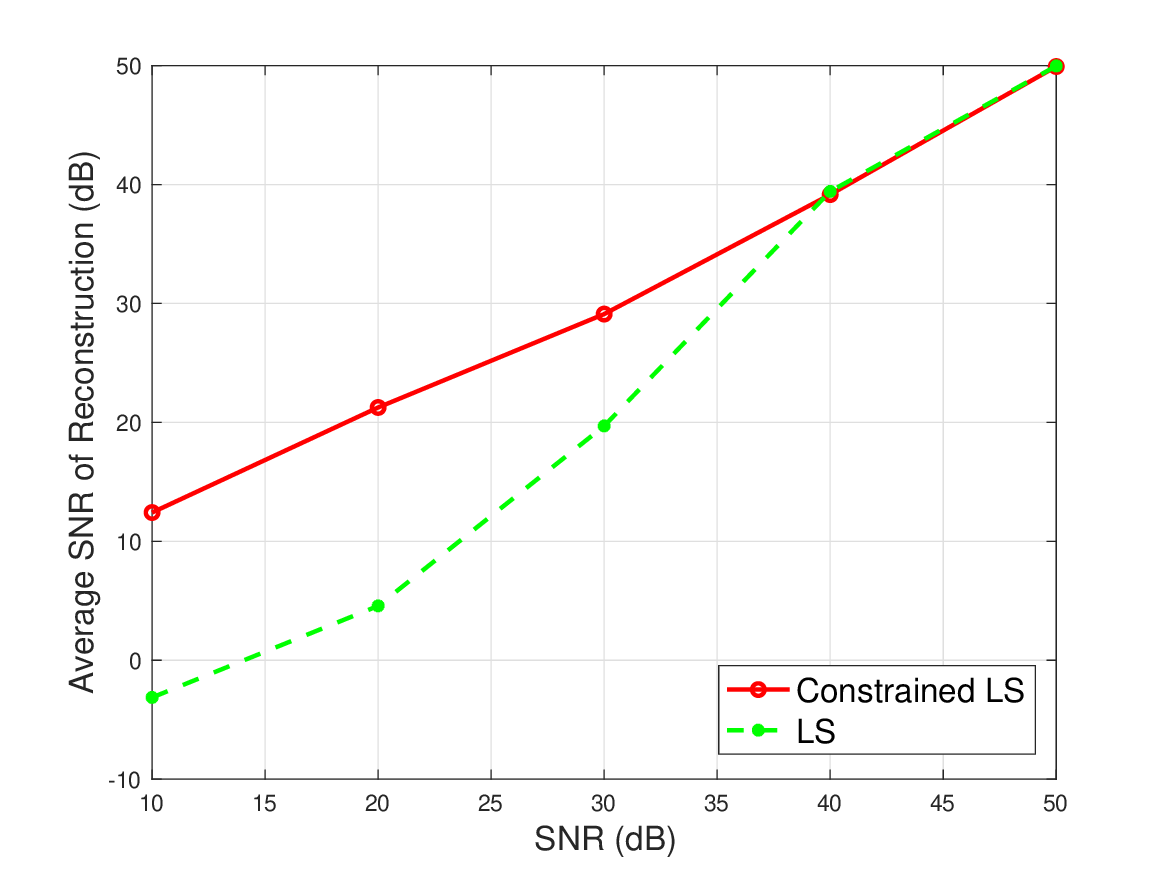}}}
	\subfloat[]{{\includegraphics[width=0.45\textwidth]{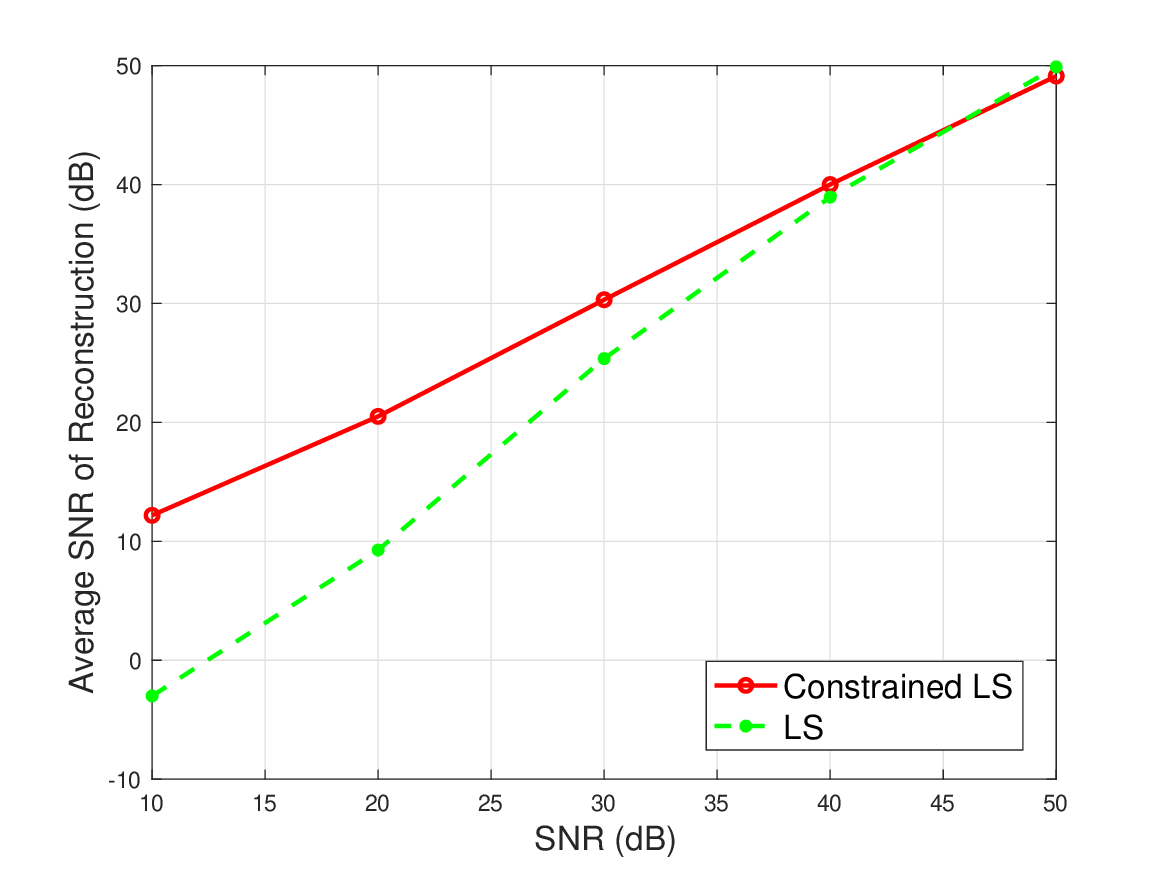}}}
	\caption{Average SNR of Reconstruction vs. Noise Level: (A) Real-Valued Case; (B) Complex-Valued Case}
	\label{fig: robust_comparison0}
\end{figure}

\subsection{Sparse Signal Case} 
In the following experiments, we assume that $\vh$ is supported on a short set, specifically with the support set $\{1, \ldots, 10\}$. The $K$-sparse signal $\vx$ has its support set drawn uniformly at random. When focusing on the PALM algorithm, we set $\lambda = 1 \times 10^{-7}$, and use the following three types of initializations: (1) $\vh_0$ and $\vx_0$ are generated according to (\ref{eqn: h_new}) and the model (\ref{eqn: LASSO_new}), abbreviated as  \textbf{Constructed PALM}; (2) all the elements in $\vh_0$ and $\vx_0$ are drawn from a standard real (or complex) Gaussian distribution,  abbreviated as  \textbf{Randomized PALM}; (3) take $\vx_0 = \boldsymbol{0}$ and $\vh_0 = [1, 0, \ldots, 0]^T$, abbreviated as \textbf{Deterministic PALM}. The maximum number of iterations for the PALM algorithm with different initialization schemes is set to $200$.

We compare the empirical success rates of various algorithms in both real and complex noiseless settings. For $K = 3$, with $L$ varying from 2 to 10, the comparison results are depicted in Figure \ref{fig: comparison2}. The figure clearly demonstrates that the Constructed PALM algorithm outperforms other algorithms, requiring fewer masks to achieve successful recovery, even when compared to PALM with randomized and deterministic initializations. Notably, we observe that complex random masks yield better performance in terms of the success rate of recovery compared to real Rademacher masks.
\begin{figure}[htbp]
	\centering
	\subfloat[]{{\includegraphics[width=0.45\textwidth]{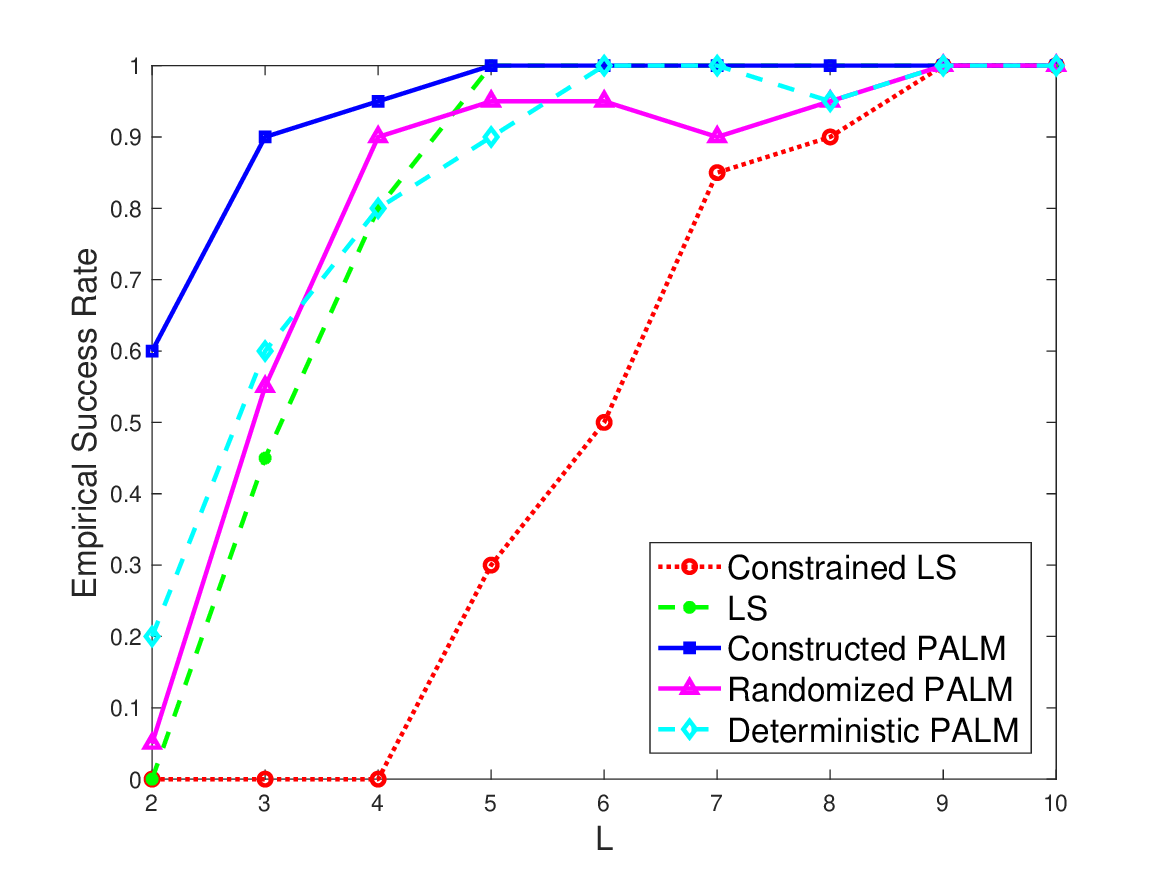}}}
	\subfloat[]{{\includegraphics[width=0.45\textwidth]{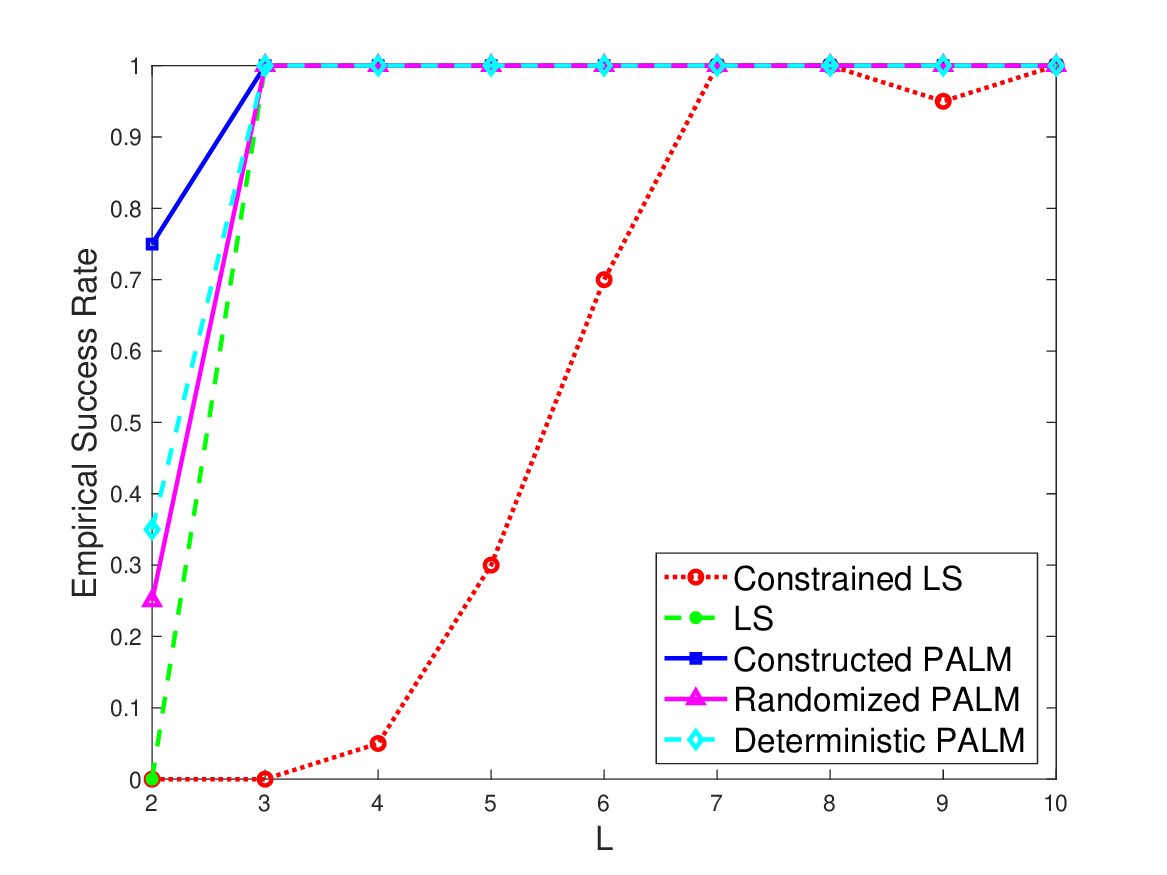}}}
	\caption{Performance Comparison of Different Models in the Sparse Scenario: (A) Real-Valued Case; (B) Complex-Valued Case.}
	\label{fig: comparison2}
\end{figure}

In the presence of noise in sparse scenarios, we select different values for $L$ based on the observed behavior in the noiseless case: $L = 6$ for the real case and $L = 5$ for the complex case. The signal-to-noise ratio (SNR) is varied from 10 dB to 50 dB. As shown in Figure \ref{fig: comparison3}, we observe that the Constructed PALM algorithm still outperforms state-of-the-art algorithms in high SNR regimes. The relatively inferior performance of Constructed PALM compared to the least squares model is primarily due to the fact that the regularization parameter $\lambda$ is not tuned when noise is introduced. Furthermore, the linear error bound for the Constructed PALM algorithm is also presented. In comparison with Randomized PALM and Deterministic PALM, our algorithm underscores the importance of good initialization in reducing the error bound in noisy settings.

\begin{figure}[htbp]
	\centering
	\subfloat[]{{\includegraphics[width=0.45\textwidth]{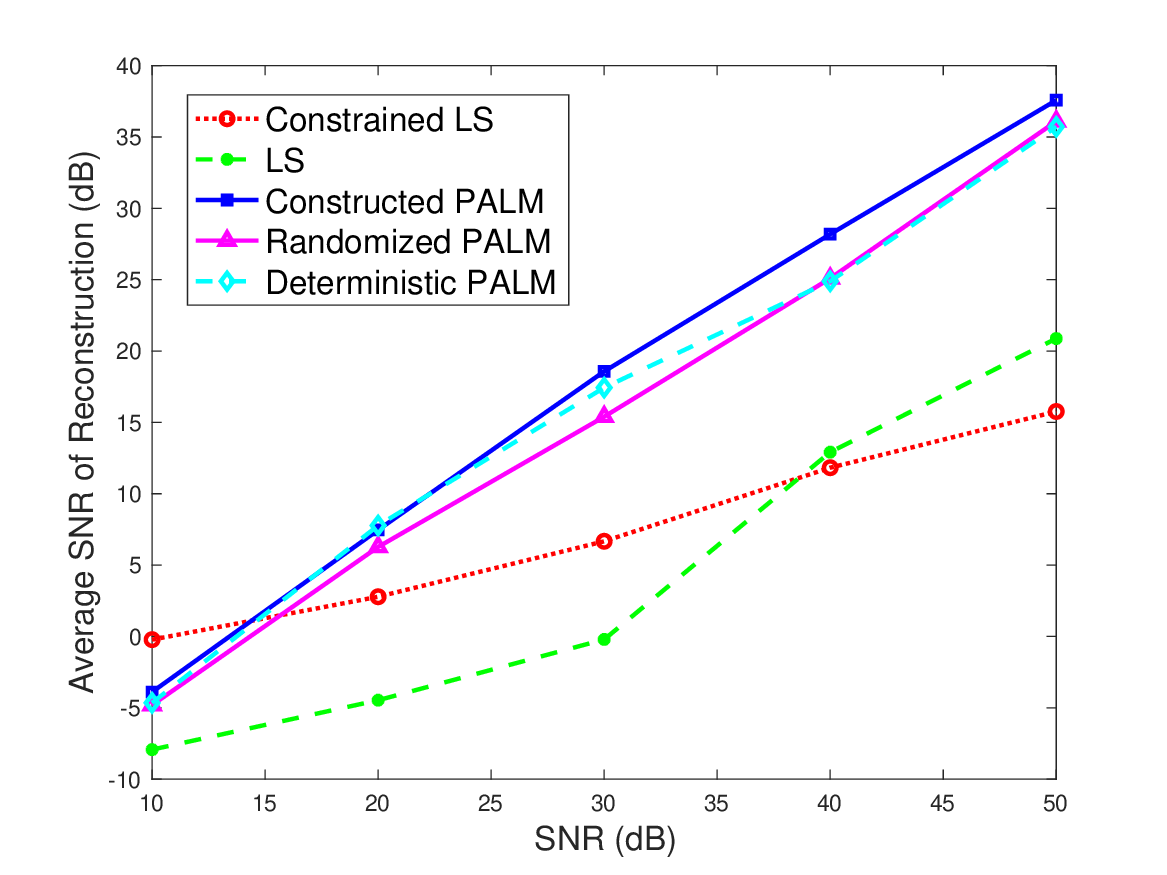}}}
	\subfloat[]{{\includegraphics[width=0.45\textwidth]{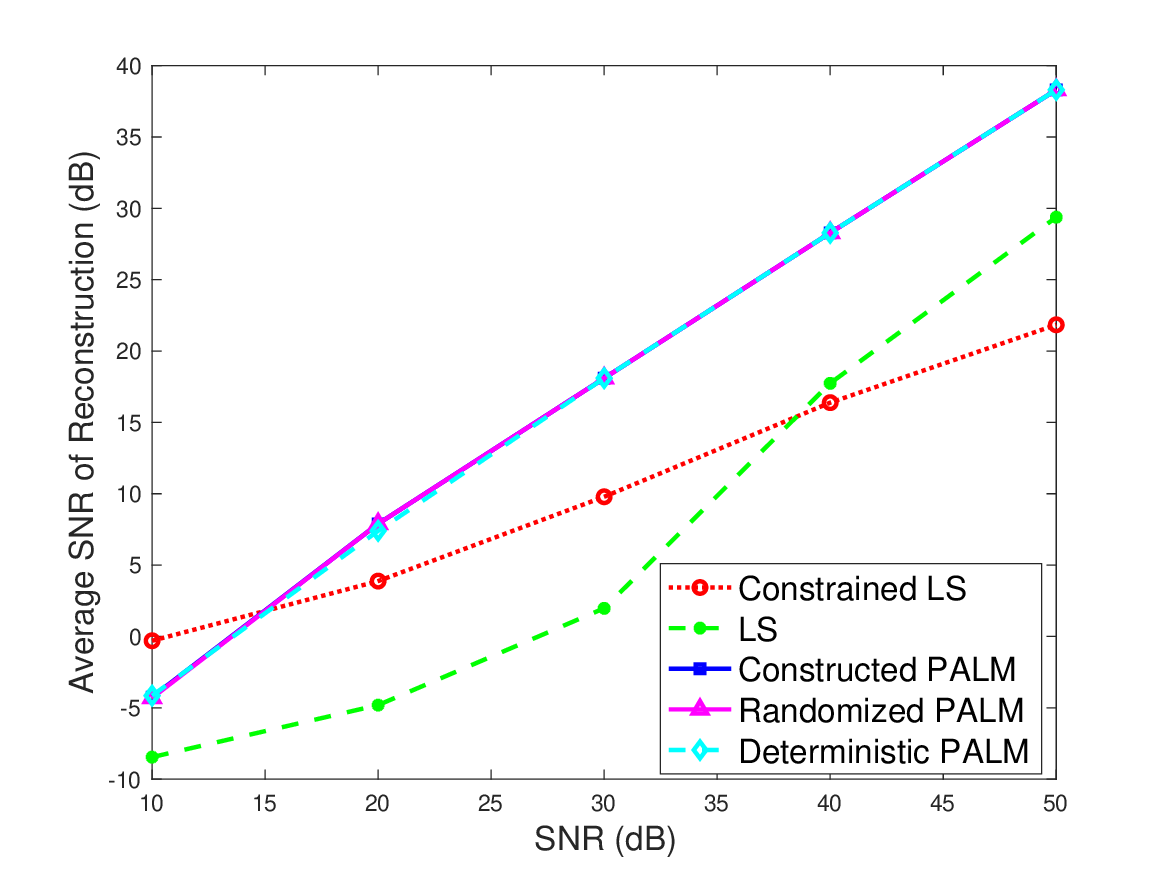}}}
	\caption{Average SNR of Reconstruction vs. Noise Level in the Sparse Scenario: (A) Real-Valued Case; (B) Complex-Valued Case.}
	\label{fig: comparison3}
\end{figure}

\subsection{Two Dimensional Case}
Furthermore, we investigate random mask blind deconvolution in a more realistic imaging scenario. Let $\boldsymbol{x}$ represent a $128 \times 128$ image, and $\boldsymbol{h}$ denote a $10 \times 10$ 2-D Gaussian filter, as shown in Figure \ref{fig: origin}, which is similar to the setup in \cite[Figure 3]{BD_LS}. Figure \ref{fig: output} presents the recovery results using the PALM algorithm and the least squares method, both with $L = 30$ Rademacher random masks, and the addition of small noise. It is evident that the PALM algorithm effectively combines the deconvolution and denoising steps, leading to superior performance compared to the least squares method.

\begin{figure}[htbp]
	\centering
	\subfloat[]{{\includegraphics[width=0.3\textwidth]{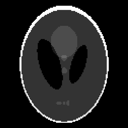}}}\qquad
	\subfloat[]{{\includegraphics[width=0.1\textwidth]{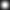}}}
	\caption{Original Image and Gaussian Filter: (a) Original Image; (b) Gaussian Filter.}
	\label{fig: origin}
\end{figure}
\begin{figure}[htbp]
	\centering
	\subfloat[]{{\includegraphics[width=0.3\textwidth]{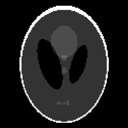}}}\qquad
	\subfloat[]{{\includegraphics[width=0.3\textwidth]{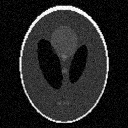}}}
	\caption{Performance Comparison of Different Models in the Sparse Scenario: (a) Recovered Image by PALM algorithm, RMSE = 82.69 dB; (b) Recovered Image by Least Squares Method, RMSE = 15.27 dB.}
	\label{fig: output}
\end{figure}

\section{Proof of Theorem \ref{thm: stable}}\label{sec: stable}
We begin by introducing fundamental technical notations and preliminary results. Let $\widehat{\boldsymbol{h}}, \boldsymbol{x} \in \mathbb{C}^n$ be fixed vectors satisfying $\|\widehat{\boldsymbol{h}}\|_2 = \|\boldsymbol{x}\|_2 = 1$. The  tangent space $T$ is defined as
 \[T=\{\widehat{\boldsymbol{h}}\overline{\boldsymbol{v}}^{*}+\boldsymbol{u}\overline{\boldsymbol{x}}^{*}\ :\ \boldsymbol{u},\boldsymbol{v}\in\mathbb{C}^{n}\}=\{\widehat{\boldsymbol{h}}\boldsymbol{v}^{T}+\boldsymbol{u}\boldsymbol{x}^{T}\ :\ \boldsymbol{u},\boldsymbol{v}\in\mathbb{C}^{n}\}.\]
The orthogonal projection onto $T$ and its complement $T^\perp$ are given by
 \begin{equation}\label{eqn: PT}
 \mathcal{P}_{T}(\boldsymbol{X})=\widehat{\boldsymbol{h}}\widehat{\boldsymbol{h}}^{*}\boldsymbol{X}+\boldsymbol{X}\overline{\boldsymbol{x}}\overline{\boldsymbol{x}}^{*}-\widehat{\boldsymbol{h}}\widehat{\boldsymbol{h}}^{*}\boldsymbol{X}\overline{\boldsymbol{x}}\overline{\boldsymbol{x}}^{*}=\widehat{\boldsymbol{h}}\widehat{\boldsymbol{h}}^{*}\boldsymbol{X}+\boldsymbol{X}\overline{\boldsymbol{x}}\boldsymbol{x}^{T}-\widehat{\boldsymbol{h}}\widehat{\boldsymbol{h}}^{*}\boldsymbol{X}\overline{\boldsymbol{x}}\boldsymbol{x}^{T},
 \end{equation}
 and
 \begin{equation}\label{eqn: PTC}
 \mathcal{P}_{T^{\perp}}(\boldsymbol{X})=\left(\boldsymbol{I}-\widehat{\boldsymbol{h}}\widehat{\boldsymbol{h}}^{*}\right)\boldsymbol{X}\left(\boldsymbol{I}-\overline{\boldsymbol{x}}\overline{\boldsymbol{x}}^{*}\right)=\left(\boldsymbol{I}-\widehat{\boldsymbol{h}}\widehat{\boldsymbol{h}}^{*}\right)\boldsymbol{X}\left(\boldsymbol{I}-\overline{\boldsymbol{x}}\boldsymbol{x}^{T}\right).
 \end{equation}
 
The following two lemmas are direct consequences of the results established in \cite{Romberg_RandomMask}. In particular, Lemma \ref{lem: RIP_convex} follows directly from Lemma 2 in \cite{Romberg_RandomMask}, while Lemma \ref{lem: dual} is derived from the proof of Theorem 2 therein. These results generalize the existing findings for the Rademacher masking model to the broader masking framework introduced in Definition \ref{def: g}. Since the underlying proof techniques remain largely analogous, we omit the detailed derivations for brevity.
\begin{lem}\cite[Lemma 2]{Romberg_RandomMask}
\label{lem: RIP_convex}
Let $\boldsymbol{D}_l=\mathrm{diag}(\vd_l)$, for $l = 1, \dots, L$, be independent diagonal matrices whose diagonal entries are independent copies of some random variable $g \in \mathbb{C}$, as defined in Definition \ref{def: g}, with parameter $\nu$. Consider the linear operator $\mathcal{A}$ as specified in (\ref{eqn: A_operator}). Then, for any $\beta > 0$, if the following condition holds:
\[
L\geq C_{\beta} C_{\nu}\mu \log^2n\log\log(n+1),  
\]
it follows that
\[
\frac{1}{2}\|\vX\|^2_F\leq \|\mathcal{A}(\vX)\|^2_F\leq \frac{3}{2}\|\vX\|^2_F
\]
for all $\vX \in T$, with probability at least $1 - 3n^{-\beta}$. Here, $C_{\nu}$ is a positive constant that depends on $\nu$.
\end{lem}
\begin{lem}\cite{Romberg_RandomMask}
\label{lem: dual}
Let $\boldsymbol{D}_l=\mathrm{diag}(\vd_l)$, for $l = 1, \dots, L$, be independent diagonal matrices whose diagonal entries are independent copies of some random variable $g \in \mathbb{C}$, as defined in Definition \ref{def: g}, with parameter $\nu$. Consider the linear operator $\mathcal{A}$ as specified in (\ref{eqn: A_operator}).  Under the given assumptions, there exists a matrix $\boldsymbol{Y} \in \mathrm{range}(\mathcal{A}^{*})$ satisfying the following conditions:
\[
\sqrt{2}\|\mathcal{A}\|\left\|\mathcal{P}_{T}(\vY)-\widehat{\vh}\vx^{T}\right\|_F\leq \frac{1}{4},
\]
and 
\[
\|\mathcal{P}_{T^{\perp}}(\vY)\|\leq \frac{1}{2},
\]
with probability at least $1-C_0n^{-\beta+1}$, provided that 
\[
L\overset{\beta}{\gtrsim} C_{\nu} \mu \log^{2}{n}\log(n/\mu)\log\log(n).
\]
 Here $C_{\nu}$ is the same constant as  in Lemma \ref{lem: RIP}, and $C_0$ is an positive absolute constant. 
\end{lem}

\begin{proof}[Proof of Theorem \ref{thm: stable}]
The proof builds upon techniques from the proof of Theorem 2 in \cite{ARR14}. However, due to the structural differences in the measurement model, a more detailed analysis is required to adapt the result to the current setting. Define $\boldsymbol{H} = \boldsymbol{X}^{\#} - \widehat{\boldsymbol{h}} \boldsymbol{x}^T$, and let $\mathcal{P}_{\mathcal{A}}$
  denote the projection operator onto the row space of $\mathcal{A}$. Explicitly, this projection is given by $\mathcal{P}_{\mathcal{A}}(\boldsymbol{H}) = \mathcal{A}^*(\mathcal{A} \mathcal{A}^*)^{-1} \mathcal{A}(\boldsymbol{H}).$ Thus $\mathcal{P}_{\mathcal{A}^{\perp}}(\vH)=\vH-\mathcal{P}_{\mathcal{A}}(\boldsymbol{H})$.

Assume that for $L \gtrsim C_{\nu} \mu \log^2 n \log(n/\mu) \log\log(n)$, the following condition is satisfied:
\begin{equation} 
\label{eqn: H_upper_temp1}
\|\mathcal{P}_{\mathcal{A}^{\perp}}(\vH)\|_F^2\leq 16(2\|\mathcal{A}\|^2+1)n\|\mathcal{P}_{\mathcal{A}}(\vH)\|_F^2,
 \end{equation}
 with  probability  at least $1 - (3 + C_0)n^{-1}$, where $C_0$ is the constant from Lemma \ref{lem: dual}. Consequently, we have
\begin{equation}
\label{eqn: H_upper1}
\begin{aligned}
\|\vH\|_F^2\leq &\|\mathcal{P}_{\mathcal{A}}(\vH)\|_F^2+\|\mathcal{P}_{\mathcal{A}^{\perp}}(\vH)\|_F^2\\
\leq & (16(2\|\mathcal{A}\|^2+1)n+1)\|\mathcal{P}_{\mathcal{A}}(\vH)\|_F^2\leq 32(2\|\mathcal{A}\|^2+1)n\|\mathcal{P}_{\mathcal{A}}(\vH)\|_F^2.
\end{aligned}
\end{equation}
Next , since $\|\mathcal{A}(\vH)\|_F\leq \|\mathcal{A}(\vX^{\#})-\widehat{\vY}\|_F+\|\mathcal{A}(\widehat{\vh}\vx^{T})-\widehat{\vY}\|_F\leq 2\|\mathcal{A}(\widehat{\vh}\vx^{T})-\widehat{\vY}\|_F\leq  2\|\widehat{\vZ}\|_F$, it follows that 
\begin{equation}
\label{eqn: PA_H_z}
\|\mathcal{P}_{\mathcal{A}}(\vH)\|_F=\|\mathcal{A}^{*}(\mathcal{A}\mathcal{A}^{*})^{-1}\mathcal{A}(\vH)\|_F\leq \|\mathcal{A}^{*}(\mathcal{A}\mathcal{A}^{*})^{-1}\|\cdot \|\mathcal{A}(\vH)\|_F\leq 2\|\mathcal{A}^{*}(\mathcal{A}\mathcal{A}^{*})^{-1}\|\cdot\|\widehat{\vZ}\|_F.
\end{equation}
Substituting this result into the inequality for $\|\boldsymbol{H}\|_F^2$ in (\ref{eqn: H_upper1}), we obtain:
\[
 \|\vH\|_F\leq 128\cdot n\cdot (2\|\mathcal{A}\|^2+1)\cdot \|\mathcal{A}^{*}(\mathcal{A}\mathcal{A}^{*})^{-1}\|^2\cdot \|\widehat{\vZ}\|_F^2
\leq 128\cdot n\cdot (2\|\mathcal{A}\|^2+1)\cdot \|\mathcal{A}\|^2\cdot \|(\mathcal{A}\mathcal{A}^{*})^{-1}\|^2\cdot  \|\widehat{\vZ}\|_F^2.
\]
This simplifies to
\begin{equation}
\label{eqn: final_temp_H}
 \|\vH\|_F\leq \frac{128\cdot n\cdot (2\lambda_{\max}(\mathcal{A}\mathcal{A}^{*})+1)\cdot \lambda_{\max}(\mathcal{A}\mathcal{A}^{*})}{\lambda^2_{\min}(\mathcal{A}\mathcal{A}^{*})}\|\widehat{\vZ}\|_F^2,
\end{equation}
where  
\[
\lambda_{\max}(\mathcal{A}\mathcal{A}^{*}):=\max_{\|\vW\|_F=1}\|\mathcal{A}\mathcal{A}^{*}(\vW)\|_F\quad \text{and}\quad \lambda_{\min}(\mathcal{A}\mathcal{A}^{*}):=\min_{\|\vW\|_F=1}\|\mathcal{A}\mathcal{A}^{*}(\vW)\|_F,
\]
and they satisfy
\[
\lambda_{\max}(\mathcal{A}\mathcal{A}^{*})=\|\mathcal{A}\mathcal{A}^*\|=\|\mathcal{A}\|^2\quad \text{and}\quad \|(\mathcal{A}\mathcal{A}^{*})^{-1}\|=\frac{1}{\lambda_{\min}(\mathcal{A}\mathcal{A}^{*})}.
\]

To derive the upper bound for (\ref{eqn: final_temp_H}), we need to further evaluate $\lambda_{\max}(\mathcal{A} \mathcal{A}^*)$  and $\lambda_{\min}(\mathcal{A} \mathcal{A}^*)$. As outlined in (\ref{eqn: A_operator}), the linear operator $\mathcal{A} : \mathbb{C}^{n \times n} \rightarrow \mathbb{C}^{n \times L}$  and its adjoint $\mathcal{A}^* : \mathbb{C}^{n \times L} \rightarrow \mathbb{C}^{n \times n}$ are defined as follows:
\[
\mathcal{A}(\vX)=\frac{1}{\sqrt{L}}(\vF\odot \vX)\vD_{g},\quad \text{and}\quad \mathcal{A}^*(\vY)=\frac{1}{\sqrt{L}}\overline{\vF}\odot (\vY\vD_{g}^*),
\]
where $\vD_g$ is defined in (\ref{eqn: Y_g}). Consequently, for any $\vY \in \mathbb{C}^{n\times L}$, the composition $\mathcal{A}\mathcal{A}^*(\vY)$ can be expressed as
\[
\mathcal{A}\mathcal{A}^*(\vY)=\frac{1}{L}\vY\vD_{g}^*\vD_{g}.
\]
From this representation, we can deduce the following:
\[
\lambda_{\max}(\mathcal{A}\mathcal{A}^*)=\frac{1}{L}\sigma_{\max}^2(\vD_g)\quad \text{and}\quad \lambda_{\min}(\mathcal{A}\mathcal{A}^*)=\frac{1}{L}\sigma_{\min}^2(\vD_g),
\]
where $\sigma_{\max}(\vD_g)$ and $\sigma_{\min}(\vD_g)$ denote the largest and smallest singular values of $\vD_g$, respectively.

Applying Theorem 4.6.1 from \cite{V11}, for any $t > 0$, we can derive the following inequalities:
\begin{equation}\label{eqn: singular}
\sqrt{n}-C\nu^2(\sqrt{L}+t)\leq \sigma_{\min}(\vD_g)\leq \sigma_{\max}(\vD_g)\leq \sqrt{n}+C\nu^2(\sqrt{L}+t),
\end{equation}
with probability at least $1 - \exp(-t^2)$, where $C$ is a positive constant. Assuming that $n \geq \max\{16 C^2 \nu^4, 1\} L$, we set $t = \frac{\sqrt{n}}{2 C \nu^2}$ in (\ref{eqn: singular}), which yields the refined estimates:
\begin{equation}
\label{eqn: singular_est}
\frac{1}{4}\sqrt{n}\leq \sigma_{\min}(\vD_g)\leq \sigma_{\max}(\vD_g)\leq \frac{7}{4}\sqrt{n},
\end{equation}
with probability  at least $1 - \exp\left(-\frac{n}{4C^2\nu^4}\right)$. Substituting the bounds from (\ref{eqn: singular_est}) into (\ref{eqn: final_temp_H}), we obtain the following:
\[
\|\boldsymbol{X}^{\#}-\widehat{\boldsymbol{h}}\boldsymbol{x}^{T}\|_F\lesssim  \sqrt{n}\|\widehat{\vZ}\|_F,
\]
which leads to the conclusion stated in (\ref{eqn: noise_est}) with probability at least
\[
1 - (3 + C_0)n^{-1} - \exp\left(-\frac{n}{4C^2 \nu^4}\right),
\]
 provided that $L \gtrsim C_{\nu} \mu \log^{2}{n} \log\left(\frac{n}{\mu}\right) \log\log(n)$ and $n \geq \max\{16 C^2 \nu^4, 1\} L$.

The remaining task is to establish the validity of (\ref{eqn: H_upper_temp1}). To do this, let us assume that for any matrix $\vX \in \text{Null}(\mathcal{A})$ (i.e., $\mathcal{A}(\vX) = \boldsymbol{0}$), with probability at least $1 - C_0 n^{-1}$, the following inequality holds:
\begin{equation}
\label{eqn: stable_temp}
\|\widehat{\vh}\vx^{T}+\vX\|_*-\|\widehat{\vh}\vx^{T}\|_*\geq \frac{1}{4}\|\mathcal{P}_{T^{\perp}}(\vX)\|_*.
\end{equation}
provided that $L{\gtrsim} C_{\nu} \mu \log^{2}{n}\log(n/\mu)\log\log(n)$. 
Since $\mathcal{P}_{\mathcal{A}^{\perp}}(\vH)\in \text{Null}(\mathcal{A})$, we can use this assumption to derive the following bound. Specifically, we have:
\begin{equation}
\label{eqn: temp_upper_thm1}
\frac{1}{4} \|\mathcal{P}_{T^\perp} \mathcal{P}_{\mathcal{A}^\perp}(\vH)\|_F \leq \|\mathcal{P}_{T^\perp} \mathcal{P}_{\mathcal{A}^\perp}(\vH)\|_* \leq \|\widehat{\vh} \vx^T + \mathcal{P}_{T^\perp} \mathcal{P}_{\mathcal{A}^\perp}(\vH)\|_* - \|\widehat{\vh} \vx^T\|_*. 
\end{equation}         
Now, we can simplify the right-hand side in (\ref{eqn: temp_upper_thm1}) by noting that:
\[
\|\widehat{\vh} \vx^T + \mathcal{P}_{T^\perp} \mathcal{P}_{\mathcal{A}^\perp}(\vH)\|_* - \|\widehat{\vh} \vx^T\|_* = \|\vX^{\#} - \vH + \mathcal{P}_{\mathcal{A}^\perp}(\vH)\|_* - \|\widehat{\vh} \vx^T\|_*. 
\]                
Applying the triangle inequality, we obtain:
\begin{equation}
\label{eqn: PTA_temp1}
\begin{aligned}
\frac{1}{4}\|\mathcal{P}_{T^{\perp}}\mathcal{P}_{\mathcal{A}^{\perp}}(\vH)\|_F\leq &
 \|\vX^{\#}-\vH+\mathcal{P}_{\mathcal{A}^{\perp}}(\vH)\|_*-\|\widehat{\vh}\vx^{T}\|_*\\
\leq &\|\mathcal{P}_{\mathcal{A}}(\vH)\|_*+\|\vX^{\#}\|_*-\|\widehat{\vh}\vx^{T}\|_*\\
\leq& \|\mathcal{P}_{\mathcal{A}}(\vH)\|_*\leq \sqrt{n}\|\mathcal{P}_{\mathcal{A}}(\vH)\|_F.
\end{aligned}
\end{equation}
The final line follows from the fact that $\|\vX^{\#}\|_* \leq R = \|\widehat{\vh} \vx^T\|_*$. Therefore, we arrive at the following result:
\begin{equation}
\label{eqn: temp1_PTA}
\begin{aligned}
\|\mathcal{P}_{\mathcal{A}^{\perp}}(\vH)\|_F^2=&\|\mathcal{P}_{T}\mathcal{P}_{\mathcal{A}^{\perp}}(\vH)\|_F^2+\|\mathcal{P}_{T^{\perp}}\mathcal{P}_{\mathcal{A}^{\perp}}(\vH)\|_F^2\overset{(a)}\leq 16(2\|\mathcal{A}\|^2+1)n\|\mathcal{P}_{\mathcal{A}}(\vH)\|_F^2, 
\end{aligned}
\end{equation}
which leads to the conclusion in equation (\ref{eqn: H_upper_temp1}).
Here $(a)$ follows from (\ref{eqn: PTA_temp1}) and 
\begin{equation}\label{eqn: PT_Z}
\frac{1}{\sqrt{2}}\|\mathcal{P}_{T}\mathcal{P}_{\mathcal{A}^{\perp}}(\vH)\|_F\overset{(b)}\leq \|\mathcal{A}(\mathcal{P}_{T}\mathcal{P}_{\mathcal{A}^{\perp}}(\vH))\|_F\overset{(c)}= \|\mathcal{A}(\mathcal{P}_{T^{\perp}}\mathcal{P}_{\mathcal{A}^{\perp}}(\vH))\|_F \leq \|\mathcal{A}\|\left\|\mathcal{P}_{T^{\perp}}\mathcal{P}_{\mathcal{A}^{\perp}}(\vH)\right\|_{F},
\end{equation}
where $(b)$ follows from Lemma \ref{lem: RIP_convex} which ensures that, with probability at least $1 - 3n^{-2}$, the frobenious norm of the projection of $\mathcal{P}_T \mathcal{P}_{\mathcal{A}^{\perp}}(\vH)$ is bounded by the norm of $\mathcal{A}$ applied to it, under the condition that $L \gtrsim C_{\nu} \mu \log^2 n \log \log (n + 1)$, and $(c)$ follows from the fact that  $\mathcal{P}_{\mathcal{A}^{\perp}}(\vH)\in \text{Null}(\mathcal{A})$.

The final step is to prove  (\ref{eqn: stable_temp}). By the definition of $\mathcal{P}_{T^{\perp}}$, we can choose $\vU_{\perp}$  and $\vV_{\perp}$ such that the matrices $[\widehat{\vh}, \vU_{\perp}]$ and $[\overline{\vx}, \vV_{\perp}]$ are unitary, and the following identity holds:
 \[
 \text{Re}(\langle \vU_{\perp}\vV_{\perp}^{*},\mathcal{P}_{T^{\perp}}(\vX)\rangle)=\|\mathcal{P}_{T^{\perp}}(\vX)\|_{*},
 \]
 where $\vX \in \text{Null}(\mathcal{A})$. Next, for any $\vX \in \text{Null}(\mathcal{A})$, let $\vY \in \text{range}(\mathcal{A}^*)$
  be chosen according to Lemma \ref{lem: dual}. With probability at least $1 - C_0 n^{-1}$, we have:
\[
\small
\begin{aligned}
\|\widehat{\vh}\vx^{T}+\vX\|_*\geq& \text{Re}(\langle \widehat{\vh}\vx^{T}+\vU_{\perp}\vV_{\perp}^{*},\widehat{\vh}\vx^{T}+\vX\rangle)\\
=&\|\widehat{\vh}\vx^{T}\|_{*}+\text{Re}(\langle \widehat{\vh}\vx^{T}+\vU_{\perp}\vV_{\perp}^{*},\vX\rangle)=\|\widehat{\vh}\vx^{T}\|_{*}+\text{Re}(\langle \widehat{\vh}\vx^{T}+\vU_{\perp}\vV_{\perp}^{*}-\vY,\vX\rangle)\\
=&\|\widehat{\vh}\vx^{T}\|_{*}+\text{Re}(\langle \widehat{\vh}\vx^{T}+\vU_{\perp}\vV_{\perp}^{*}-\mathcal{P}_{T}(\vY),\mathcal{P}_{T}(\vX)\rangle)+\text{Re}(\langle \widehat{\vh}\vx^{T}+\vU_{\perp}\vV_{\perp}^{*}-\mathcal{P}_{T^{\perp}}(\vY),\mathcal{P}_{T^{\perp}}(\vX)\rangle)\\
=&\|\widehat{\vh}\vx^{T}\|_{*}+\|\mathcal{P}_{T^{\perp}}(\vX)\|_*+\text{Re}(\langle \widehat{\vh}\vx^{T}-\mathcal{P}_{T}(\vY),\mathcal{P}_{T}(\vX)\rangle)-\text{Re}(\langle \mathcal{P}_{T^{\perp}}(\vY),\mathcal{P}_{T^{\perp}}(\vX)\rangle)\\
\geq &\|\widehat{\vh}\vx^{T}\|_{*}+ (1-\|\mathcal{P}_{T^{\perp}}(\vY)\|)\cdot\|\mathcal{P}_{T^{\perp}}(\vX)\|_*-\|\widehat{\vh}\vx^{T}-\mathcal{P}_{T}(\vY)\|_F\cdot \|\mathcal{P}_{T}(\vX)\|_F\\
\overset{(d)}\geq &\|\widehat{\vh}\vx^{T}\|_{*}+(1-\|\mathcal{P}_{T^{\perp}}(\vY)\|)\cdot\|\mathcal{P}_{T^{\perp}}(\vX)\|_*-\sqrt{2}\|\widehat{\vh}\vx^{T}-\mathcal{P}_{T}(\vY)\|_F\cdot\|\mathcal{A}\|\cdot \|\mathcal{P}_{T^{\perp}}(\vX)\|_* \\
\overset{(e)}\geq & \|\widehat{\vh}\vx^{T}\|_{*}+(1-1/2-1/4)\cdot\|\mathcal{P}_{T^{\perp}}(\vX)\|_*=\|\widehat{\vh}\vx^{T}\|_{*}+\frac{1}{4}\|\mathcal{P}_{T^{\perp}}(\vX)\|_*.
\end{aligned}
\]
provided that $L \gtrsim C_{\nu} \mu \log^2 n \log(n/\mu) \log \log n$. This completes the proof and leads directly to the result in (\ref{eqn: stable_temp}). Here,  $(d)$ follows the same lines as in (\ref{eqn: PT_Z}) by replacing $\mathcal{P}_{\mathcal{A}^{\perp}}(\vH)$ into $\vX$, and $(e)$ follows from Lemma \ref{lem: dual}.

\end{proof}

\section{Proof of Theorem \ref{thm: lower_bound}}\label{sec: Thm2}
We begin by introducing the matrix Bernstein inequality, which serves as a fundamental tool in our analysis.
\begin{thm}
\label{thm: Bernstein_inequality}\cite[Theorem 1.6]{T12} (Matrix
Bernstein's Inequality) Let $\{ \boldsymbol{Z}_k \}$  be a finite sequence of independent, random matrices of dimension $d_1 \times d_2$. Suppose that each matrix satisfies the moment condition:
 \[
\mathbb{E}\boldsymbol{Z}_{k}=\boldsymbol{0}\qquad\text{and}\qquad\|\boldsymbol{Z}_{k}\|\leq R,\qquad \text{almost} \ \text{surely}.
\]
Define 
\[
\sigma^{2}:=\max\left\{ {\Big\Vert \sum_{k}\mathbb{E}(\boldsymbol{Z}_{k}\boldsymbol{Z}_{k}^{*})\Big\Vert ,\Big\Vert \sum_{k}\mathbb{E}(\boldsymbol{Z}_{k}^{*}\boldsymbol{Z}_{k})\Big\Vert }\right\} .
\]
Then, for any $t \geq 0$, the following concentration bound holds:
\[
\mathbb{P}\left\{ \Big\|\sum_{k}\boldsymbol{Z}_{k}\Big\|\geq t\right\} \leq(d_{1}+d_{2})\exp\left(\frac{-t^{2}/2}{\sigma^{2}+Rt/3}\right).
\]
\end{thm}
Next, we establish two key technical lemmas that are instrumental in the proof of Theorem \ref{thm: lower_bound}.
\begin{lem}\label{lem: bernstein}
Let $\boldsymbol{d}_l \in \mathbb{C}^n$, for $l=1, \dots, L$, be independent vectors whose entries are independent copies of the random variable $g$ as defined in Definition \ref{def: g} with parameter $\nu$. Then, with probability at least $1 - \frac{2}{n}$, the following bound holds:
\[
\|\vd_{l}\|_2\geq n-\frac{n}{2\nu^2},\quad \text{for all}\ \  l=1,\ldots,L,
\]
provided that $n\geq 32\nu^6L\log n$.
\end{lem}

\begin{proof}
For any fixed $l = 1, \ldots, L$, we express the squared norm as:
 \[
 \|\boldsymbol{d}_{l}\|_2^2 = \sum_{k=1}^n |d_{l,k}|^2,
 \]
  where $\boldsymbol{d}_{l} = [d_{l,1}, \ldots, d_{l,n}]^{T}$. Define the  random variables $z_{l,k}: = |d_{l,k}|^2 - \mathbb{E}|d_{l,k}|^2$. By such construction, we have $\mathbb{E}[z_{l,k}] = 0$, and it follows that $\max_{1\leq k\leq n}|z_{l,k}|\leq 2\nu^2$ and $\sum_{k=1}^n\mathbb{E}|z_{l,k}|^2\leq \nu^2n$. Applying Theorem \ref{thm: Bernstein_inequality} to the sequence $\{z_{l,k}\}_{k=1}^n$, and setting $R = 2\nu^2$  and $\sigma^2 \leq \nu^2 n$, we obtain the following concentration bound: for any fixed $l=1,\ldots,L$, 
\begin{equation}
\label{eqn: upper_berstein_temp}
\mathbb{P}\left\{\big|\|\vd_l\|_2^2-n\big|\geq t\right\}=\mathbb{P}\left\{|\|\vd_l\|_2^2-\mathbb{E}\|\vd_l\|_2^2|\geq t\right\}\leq 2\exp\left(\frac{-t^2/2}{\sigma^2+Rt/3}\right)\leq 2\exp\left(\frac{-t^2/2}{\nu^2n+2\nu^2t/3}\right).
\end{equation}
Substituting $t = \frac{n}{2\nu^2}$  into \eqref{eqn: upper_berstein_temp}, we obtain:
\[
\mathbb{P}\left\{\big|\|\vd_l\|_2^2-n\big|\geq \frac{n}{2\nu^2}\right\}\leq  2\exp\left(\frac{-n^2/(8\nu^4)}{\nu^2n+n/3}\right)\leq 2\exp\left(\frac{-n}{16\nu^6}\right)\leq \frac{2}{n^2},
\]
provided that $n \geq 32\nu^6L \log n$. Finally, applying the union probability  bound over all $l = 1, \dots, L$, we conclude that:
\[
\big|\|\vd_{l}\|_2^2-n\big|<\frac{n}{2\nu^2},\quad \text{for}\ \  l=1,\ldots,L,
\]
with probability at least $1 - \frac{2}{n}$, which establishes the desired result.
\end{proof}
The following lemma is a direct consequence of Hoeffding's inequality presented in \cite{V11}. Specifically, let $z_1, \dots, z_n$ be independent random variables such that $a_k \leq z_k \leq b_k$ for each $k = 1, \dots, n$, and let $\mathbb{E}[z_k] = 0$ for all $k$. Then, for any $t > 0$, we have the following concentration inequality:
\[
\mathbb{P}\left(\Big|\sum_{k=1}^nx_k\Big|\geq t\right)\leq 2\exp\left(-\frac{t^2}{\sum_{k=1}^{n}(b_k-a_k)^2}\right). 
\]
\begin{lem}\label{lem: hoeffding}
Let $\boldsymbol{d}_l \in \mathbb{C}^n$, for $l=1, \dots, L$, be independent vectors whose entries are independent copies of the random variable $g$ as defined in Definition \ref{def: g} with parameter $\nu$. Let $\boldsymbol{z} \in \mathbb{C}^n$ be an arbitrary fixed vector. Then, with probability at least $1 - \frac{1}{n}$, we have:
\[
|\boldsymbol{d}_{l}^{T}\boldsymbol{z}| \leq \nu \sqrt{2 \log n} \|\boldsymbol{z}\|_{2},\quad \text{for all}\ \  l=1,\ldots,L.
\]
\end{lem}

Now we begin to prove Theorem \ref{thm: lower_bound}.
\begin{proof}[Proof of Theorem \ref{thm: lower_bound}]
First, it is straightforward to observe that for almost all $\vx \in \mathbb{S}^{n-1}$, the following condition holds:
\begin{equation}
\label{eqn: robust_cond}
\mathrm{diag}({\vf}_j)\overline{\vD_g}(\vD_g^{T}\overline{\vD_g})^{-1}\vD^{T}_g\mathrm{diag}(\overline{\vf}_j)\vx\neq \vx, \quad j=1,\ldots,L,
\end{equation}
where $\vf_j \in \mathbb{C}^n$ denotes the $j$-th row of the matrix $\vF$. The subsequent discussion is based on the assumption that $\vx$ satisfies the condition (\ref{eqn: robust_cond}). 

\textbf{Step 1: Demonstrate that, in order to attain the conclusion stated in (\ref{eqn: error_lower}), it is sufficient to construct $\vX_0 \in \mathbb{C}^{n \times n}$  such that:
\begin{equation}\label{eqn: cond1}
-\text{Re}(\langle \widehat{\boldsymbol{h}}\boldsymbol{x}^{T},\boldsymbol{X}_0\rangle)\geq \|\mathcal{P}_{T^{\perp}}(\boldsymbol{X}_0)\|_{*},
\end{equation}
 and 
 \begin{equation}\label{eqn: cond2}
\sqrt{n}\|\mathcal{A}(\boldsymbol{X}_0)\|_F\leq  8\sqrt{2}\nu\sqrt{L\log n}\|\boldsymbol{X}_0\|_F.
\end{equation}
}

Lemma 4.1 in \cite{KS} establishes the following relationship:
\begin{equation}
\label{eqn: closure}
\overline{\mathcal{K}_{*}(\widehat{\vh}\vx^{T})}=\left\{\vX\in \mathbb{C}^{n\times n}\ :\ -\text{Re}(\langle \widehat{\vh}\vx^{T},\vX\rangle)\geq \|\mathcal{P}_{T^{\perp}}(\vX)\|_{*}\right\}.
\end{equation}
Here, $\overline{\mathcal{K}_{*}(\widehat{\boldsymbol{h}}\boldsymbol{x}^{T})}$ denotes the topological closure of $\mathcal{K}_{*}(\widehat{\boldsymbol{h}}\boldsymbol{x}^{T})$, which is defined as:
\[
\mathcal{K}_{*}(\widehat{\vh}\vx^{T})=\{\vX\in \mathbb{C}^{n\times n}\ :\ \|\widehat{\vh}\vx^{T}+\epsilon\vX\|_*\leq \|\widehat{\vh}\vx^{T}\|_*\ \text{for some}\ \epsilon >0\}.
\]
Next, assume the existence of a matrix $\boldsymbol{X}_0 \in \mathbb{C}^{n \times n}$ that satisfies  (\ref{eqn: cond1}) and (\ref{eqn: cond2}). 
By leveraging the closure property in equation (\ref{eqn: closure}) and the continuity of the operator $\mathcal{A}$, we can guarantee the existence of a scalar $t > 0$ and a matrix $\widetilde{\boldsymbol{X}} \in \mathbb{C}^{n \times n}$ such that:
\begin{equation}
\label{eqn: X_temp_final1}
\|\widehat{\boldsymbol{h}}\boldsymbol{x}^{T}+t\widetilde{\boldsymbol{X}}\|_{*}\leq \|\widehat{\boldsymbol{h}}\boldsymbol{x}^{T}\|_{*} \quad\text{and}\quad \sqrt{n}\|\mathcal{A}(\widetilde{\boldsymbol{X}})\|_F\leq 16\sqrt{2}\nu\sqrt{L\log n}\|\widetilde{\boldsymbol{X}}\|_F.
\end{equation}
Now, define the noise term $\widehat{\vZ} \in \mathbb{C}^{n \times L}$ as $\widehat{\vZ}= \mathcal{A}(t \widetilde{\boldsymbol{X}})$, and let $\widehat{\boldsymbol{Y}} = \mathcal{A}(\widehat{\boldsymbol{h}} \boldsymbol{x}^T) + \widehat{\boldsymbol{Z}}$. Consequently, an optimal solution to (\ref{eqn: noise1}) is given by ${\vX}^{\#} = \widehat{\boldsymbol{h}}\boldsymbol{x}^{T} + t\widetilde{\boldsymbol{X}}$, as it satisfies $\|\mathcal{A}(\boldsymbol{X}^{\#}) - \widehat{\boldsymbol{Y}}\|_F = 0$ and $\|\boldsymbol{X}^{\#}\|_{*} \leq \|\widehat{\boldsymbol{h}}\boldsymbol{x}^{T}\|_{*}$ as implied by the first part of (\ref{eqn: X_temp_final1}). 

Considering the second part of (\ref{eqn: X_temp_final1}) and the condition that $L=\Theta\left(C_{\nu} \mu \log^{2}{n} \log\left(\frac{n}{\mu}\right) \log\log(n)\right)$, the derived error bound leads to the following assertion as in  (\ref{eqn: error_lower}):
\[
\|\boldsymbol{X}^{\#}-\widehat{\boldsymbol{h}}\boldsymbol{x}^{T}\|_F=\|t\widetilde{\vX}\|_F\geq {\frac{\sqrt{n}}{{16\sqrt{2}\nu\sqrt{L\log n}}}}\|\mathcal{A}(t\widetilde{\vX})\|_F\gtrsim{\frac{\sqrt{n}}{{\widetilde{C}_{\nu}\sqrt{\mu}\log^3 n}}}\|\widehat{\vZ}\|_F.
\]

At this point, our primary goal is to construct an appropriate matrix $\boldsymbol{X}_0 \in \mathbb{C}^{n \times n}$  that satisfies both the conditions in (\ref{eqn: cond1}) and (\ref{eqn: cond2}).

 Without lose of generality, assume that $\widehat{{h}}_1\neq 0$.  Let $\mathcal{D}$ denote the subspace  $\mathcal{D}:=\text{span}\{\overline{\boldsymbol{d}_1},\cdots,\overline{\boldsymbol{d}_L}\}$. We define $\mathcal{P}_{\mathcal{D}}(\cdot)$ as the projection operator onto the subspace $\mathcal{D}$, and denote  
 \[
 \boldsymbol{x}_{\mathcal{D}}^{\perp}:=\boldsymbol{x}-\mathcal{P}_{\mathcal{D}}(\boldsymbol{x})=\vx-\overline{\vD_{g}}({\vD}_g^{T}\overline{\vD_g})^{-1}({\vD}_g)^{T}\vx.
 \] 
By the assumption in (\ref{eqn: robust_cond}), it follows that $\boldsymbol{x}_{\mathcal{D}}^{\perp} \neq \boldsymbol{0}$.

\textbf{Step 2: Define $\boldsymbol{W}:=-\frac{\widehat{{h}}_1}{\|\boldsymbol{x}_{\mathcal{D}}^{\perp}\|_2\cdot|\widehat{{h}}_1|}\boldsymbol{e}_1(\boldsymbol{x}_{\mathcal{D}}^{\perp})^{T}$ and show that $\|\mathcal{P}_{T^{\perp}}(\boldsymbol{W})\|_{*}\leq  \frac{\nu\sqrt{2L\log n}}{\sqrt{n-\frac{n}{2\nu^2}}}$, where $\ve_1=[1,0,\ldots,0]^{T}\in \mathbb{C}^n$. }

Using the operator $\mathcal{A}$ as defined in (\ref{eqn: A_operator}), we obtain that
\[
\begin{aligned}
\mathcal{A}(\boldsymbol{W})=&-\frac{\widehat{{h}}_1}{\|\boldsymbol{x}_{\mathcal{D}}^{\perp}\|_2\cdot|\widehat{{h}}_1|}\cdot \frac{1}{\sqrt{L}}(\vF\odot (\boldsymbol{e}_1(\boldsymbol{x}_{\mathcal{D}}^{\perp})^{T}))\vD_{g}=-\frac{\widehat{{h}}_1}{\|\boldsymbol{x}_{\mathcal{D}}^{\perp}\|_2\cdot|\widehat{{h}}_1|}\cdot \frac{1}{\sqrt{L}} \boldsymbol{e}_1(\boldsymbol{x}_{\mathcal{D}}^{\perp})^{T}\vD_{g}\\
=&-\frac{\widehat{{h}}_1}{\|\boldsymbol{x}_{\mathcal{D}}^{\perp}\|_2\cdot|\widehat{{h}}_1|}\cdot \frac{1}{\sqrt{L}} \ve_1(\vx^{T}-\vx^{T}\vD_{g}({\vD}^{*}_g\vD_g)^{-1}{\vD}_g^*)\vD_g=\boldsymbol{0}.
\end{aligned}
\]
Futhermore,  with probability at least $1-\frac{3}{n}$, we have the following estimate for the nuclear norm of the projection of $\boldsymbol{W}$ onto the orthogonal complement of $T$:
\begin{equation}\label{eqn: PTW}
\begin{split}
\|\mathcal{P}_{T^{\perp}}(\boldsymbol{W})\|_{*}=&\left\|\mathcal{P}_{T^{\perp}}\left(\frac{\boldsymbol{e}_1(\boldsymbol{x}_{\mathcal{D}}^{\perp})^{T}}{\|\boldsymbol{x}_{\mathcal{D}}^{\perp}\|_2}\right)\right\|_{*}=\left\|\left(\boldsymbol{I}-\widehat{\boldsymbol{h}}\widehat{\boldsymbol{h}}^{*}\right)\frac{\boldsymbol{e}_1(\boldsymbol{x}_{\mathcal{D}}^{\perp})^{T}}{\|\boldsymbol{x}_{\mathcal{D}}^{\perp}\|_2}\left(\boldsymbol{I}-\overline{\boldsymbol{x}}\boldsymbol{x}^{T}\right)\right\|_{*}\\
=&\left\|(\boldsymbol{I}-\widehat{\vh}\widehat{\vh}^{*})\ve_1\right\|_2\cdot \left\|\left(\boldsymbol{I}-{\boldsymbol{x}}\boldsymbol{x}^{*}\right)\frac{\boldsymbol{x}_{\mathcal{D}}^{\perp}}{\|\boldsymbol{x}_{\mathcal{D}}^{\perp}\|_2}\right\|_{2}\leq\left\|\left(\boldsymbol{I}-{\boldsymbol{x}}\boldsymbol{x}^{*}\right)\frac{\boldsymbol{x}_{\mathcal{D}}^{\perp}}{\|\boldsymbol{x}_{\mathcal{D}}^{\perp}\|_2}\right\|_{2}\\
=&\sqrt{1-\frac{|\langle \vx,\vx^{\perp}_{\mathcal{D}}\rangle|^2}{\|\vx^{\perp}_{\mathcal{D}}\|_2}}=\sqrt{1-\|\vx^{\perp}_{\mathcal{D}}\|^2_2}=\|\mathcal{P}_{\mathcal{D}}(\vx)\|_2\\
\overset{(a)}\leq &\sqrt{\sum_{l=1}^L\frac{|\vd_l^{T}\vx|^2}{\|\vd_l\|_2^2}}\overset{(b)}\leq \frac{\nu\sqrt{2L\log n}}{\sqrt{n-\frac{n}{2\nu^2}}},
\end{split}
\end{equation}
provided that $n \geq 32\nu^6 L \log n$. Here $(a)$ follows from the subspace projection property, which asserts that
\[
\|\mathcal{P}_{\mathcal{D}}(\vx)\|_2^2\leq \sum_{l=1}^L\|\mathcal{P}_{\overline{d_l}}(\vx)\|^2_2=\sum_{l=1}^L\frac{|\vd_l^{T}\vx|^2}{\|\vd_l\|_2^2},
\]
where $\mathcal{P}_{\overline{\vd_l}}(\vx)$ represents the projection of $\vx$ onto $\text{span}(\overline{\vd_l})$, while  $(b)$ follows straightforwardly from the findings presented in Lemma \ref{lem: bernstein} and Lemma \ref{lem: hoeffding}.

\textbf{Step 3: Take $\boldsymbol{X}_0:=-\beta\widehat{\boldsymbol{h}}\boldsymbol{x}^{T}+\boldsymbol{W}$, where  $\beta=\frac{2\nu\sqrt{L\log n}}{\sqrt{n-\frac{n}{2\nu^2}}}$,  and prove that it meets the conditions in (\ref{eqn: cond1}) and (\ref{eqn: cond2}).}

Based on  (\ref{eqn: PTW}), we can observe the following:
\[
-\text{Re}(\langle \widehat{\boldsymbol{h}}\boldsymbol{x}^{T},\boldsymbol{X}_0\rangle)=\beta+|\widehat{h}_1|\cdot\|\boldsymbol{x}_{\mathcal{D}}^{\perp}\|_2\geq \beta\geq  \|\mathcal{P}_{T^{\perp}}(\boldsymbol{W})\|_{*}=\|\mathcal{P}_{T^{\perp}}(\boldsymbol{X}_0)\|_{*},
\]
which satisfies the condition outlined in  (\ref{eqn: cond1}). 

Next, we consider the following expression:
\begin{equation}
\label{eqn: Ax_bernstein}
\begin{aligned}
\|\mathcal{A}(\widehat{\boldsymbol{h}}\boldsymbol{x}^{T})\|_{F}^{2}-\|\widehat{\boldsymbol{h}}\boldsymbol{x}^{T}\|_{F}^{2}=& \frac{1}{L}\|(\vF\odot(\widehat{\vh}\vx^{T}))\vD_g\|_F^2-\|\widehat{\vh}\vx^{T}\|_F^2=\frac{1}{L}\|(\vF\odot(\widehat{\vh}\vx^{T}))\vD_g\|_F^2-\|\vF\odot \widehat{\vh}\vx^{T}\|_F^2\\
=&\sum_{l=1}^L\left(\frac{1}{L}\|(\vF\odot(\widehat{\vh}\vx^{T}))\vd_l\|_2^2-\frac{1}{L}\|\vF\odot \widehat{\vh}\vx^{T}\|_F^2\right),
\end{aligned}
\end{equation}
which is the key expression we wish to analyze.

For $l=1\ldots,L$,  we have the following bound:
\[
\|(\vF\odot(\widehat{\vh}\vx^{T}))\vd_l\|_2^2=\left|\sum_{k=1}^{n}\Big|\widehat{h}_{k}\Big|^{2}\boldsymbol{x}^{*}\overline{\boldsymbol{D}}_{l}\overline{\vf_k}\vf_k^{T}\boldsymbol{D}_{l}\boldsymbol{x}\right|\leq \|\widehat{\vh}\|^2_{\infty}\Big\|\overline{\vD}_l\sum_{k=1}^n\overline{\vf}_{k}\vf^{T}_{k}\vD_l\Big\|\leq n\|\widehat{\vh}\|_\infty^2\nu^2=\mu\nu^2,
\]
using the fact that $\|\vx\|_2=1$, $n\|\widehat{\vh}\|^2_\infty=\mu\geq 1$ and $\nu\geq 1$. Thus, we conclude that 
\begin{equation}
\label{eqn: R}
\begin{aligned}
R:=&\max_{1\leq l\leq L}\left|\frac{1}{L}\|(\vF\odot(\widehat{\vh}\vx^{T}))\vd_l\|_2^2-\frac{1}{L}\|\vF\odot \widehat{\vh}\vx^{T}\|_F^2\right|
\leq& \frac{2\mu\nu^2}{L},
\end{aligned}
\end{equation}
and
\begin{equation}
\label{eqn: sigma}
\begin{aligned}
\sigma^2:=&\sum_{l=1}^{L}\mathbb{E}\left|\frac{1}{L}\|(\vF\odot(\widehat{\vh}\vx^{T}))\vd_l\|_2^2-\frac{1}{L}\|\vF\odot \widehat{\vh}\vx^{T}\|_F^2\right|^2\leq \sum_{l=1}^{L}\mathbb{E}\left|\frac{1}{L}\|(\vF\odot(\widehat{\vh}\vx^{T}))\vd_l\|_2^2\right|^2\\
\leq & \frac{\mu\nu^2}{L}\mathbb{E}\|(\vF\odot(\widehat{\vh}\vx^{T})\vd_1\|_2^2
\leq \frac{\mu\nu^2}{L}.
\end{aligned}
\end{equation}
We now apply Bernstein's inequality (Theorem \ref{thm: Bernstein_inequality}) to the expression in equation (\ref{eqn: Ax_bernstein}) along with the bounds on $R$ and $\sigma^2$ as specified in (\ref{eqn: R}) and (\ref{eqn: sigma}). This yields the following probability bound:
\begin{equation}
\label{eqn: Ahx_bernstein_result}
\mathbb{P}\left(\left|\|\mathcal{A}(\widehat{\boldsymbol{h}}\boldsymbol{x}^{T})\|_{F}^{2}-\|\widehat{\boldsymbol{h}}\boldsymbol{x}^{T}\|_{F}^{2}\right|\geq t\right)\leq 2\exp\left(\frac{-t^{2}/2}{\sigma^{2}+Rt/3}\right)\leq 2\exp\left(-\frac{Lt^2/2}{\mu\nu^2+2\mu\nu^2 t/3}\right).
\end{equation}
By setting $t=1$ in (\ref{eqn: Ahx_bernstein_result}), we obtain:
\[
\|\mathcal{A}(\widehat{\boldsymbol{h}}\boldsymbol{x}^{T})\|_{F}^{2}\leq \|\widehat{\boldsymbol{h}}\boldsymbol{x}^{T}\|_{F}^{2}+1=2\|\widehat{\boldsymbol{h}}\boldsymbol{x}^{T}\|_{F}^{2},
\]
which holds with probability at least $1 - \frac{1}{n}$, given that $L \geq 4 \mu \nu^2 \log n$.  Therefore,
\[
\begin{aligned}
\|\mathcal{A}(\boldsymbol{X}_0)\|_F=&\beta\|\mathcal{A}(\widehat{\boldsymbol{h}}\boldsymbol{x}^{T})\|_F\leq 2\beta\|\widehat{\boldsymbol{h}}\boldsymbol{x}^{T}\|_F=\frac{4\nu\sqrt{L\log n}}{\sqrt{n-\frac{n}{2\nu^2}}}\leq \frac{4\sqrt{2}\nu\sqrt{L\log n}}{\sqrt{n}}\\
 \overset{(c)}\leq&\frac{8\sqrt{2}\nu\sqrt{L\log n}}{\sqrt{n}}\|\boldsymbol{X}_0\|_F.
\end{aligned}
\]
Here $(c)$ follows from 
\[
\|\boldsymbol{X}_0\|_F\geq \|\boldsymbol{W}\|_F-\beta=1-\frac{2\nu\sqrt{L\log n}}{\sqrt{n-\frac{n}{2\nu^2}}}\geq 1-\frac{2\nu\sqrt{L\log n}}{\sqrt{\frac{n}{2}}}\geq \frac{1}{2},
\]
provided that $n \geq 32 \nu^6 L \log n$ and $\nu \geq 1$. Consequently, the condition in (\ref{eqn: cond2}) holds true.
\end{proof}

\section{Proof of Theorem \ref{thm: h_recovery}}\label{sec: Thm3}
\begin{proof}[Proof of Theorem \ref{thm: h_recovery}]
Let $\text{supp}(\boldsymbol{x}) = \{i_1, i_2, \dots, i_K\}$, and define the vectors 
\[
\boldsymbol{x} = [x_1, \dots, x_n]^T, \quad \boldsymbol{d}_l = [d_{l,1}, \dots, d_{l,n}]^T, \quad l = 1, \dots, L.
\]
 By direct calculation, the matrix $\boldsymbol{H}$ can be rewritten as
\[
\boldsymbol{H}=\frac{1}{L}\sum_{l=1}^L\sum_{k=1}^Kx_{i_k}d_{l,i_k}\overline{\boldsymbol{D}_l}\check{\vC}_{s_{i_k-1}(\boldsymbol{h})}+\frac{1}{{L}}\sum_{l=1}^{L}\overline{\boldsymbol{D}_l}\check{\vC}_{\vz_l}.
\]
Here, the cyclic shift  $s_{\tau}(\vz)$ is defined as in (\ref{eqn: s_tau}). 

For any fixed $j = 1, \ldots, n$, the $j$-th row vector $\vh_j\in \mathbb{C}^n$ of $\vH$ can be represented as:
\begin{equation}
\label{eqn: h_j}
\begin{aligned}
\vh_j&=\frac{1}{L}\sum_{l=1}^L\sum_{k=1}^Kx_{i_k}d_{l,i_k}\overline{d_{l,j}}s_{i_k-j}(\boldsymbol{h})+\frac{1}{{L}}\sum_{l=1}^{L}\overline{d_{l,j}}s_{-j+1}(\boldsymbol{\vz}_l)\\
&=x_{j}\boldsymbol{h}+\frac{1}{L}x_{j}\sum_{l=1}^L\left(|d_{l,j}|^2-1\right)\boldsymbol{h}+\frac{1}{L}\sum_{l=1}^L\sum_{i_k\neq j}x_{i_k}d_{l,i_k}\overline{d_{l,j}}s_{i_k-j}(\boldsymbol{h})+\frac{1}{{L}}\sum_{l=1}^{L}\overline{d_{l,j}}s_{-j+1}(\boldsymbol{\vz}_l)\\
&=x_{j}\boldsymbol{h}+\boldsymbol{m}_{j,0}+\sum_{l=1}^L\boldsymbol{m}_{j,l}+\sum_{l=1}^L\widetilde{\vm}_{j,l},
\end{aligned}
\end{equation}
where $\boldsymbol{m}_{j,0}:=\frac{1}{L}x_{j}\sum_{l=1}^L\left(|d_{l,j}|^2-1\right)\boldsymbol{h}$, 
\begin{equation}
\label{eqn: M}
\  \boldsymbol{m}_{j,l}:=\frac{1}{L}\sum_{i_k\neq j}x_{i_k}d_{l,i_k}\overline{d_{l,j}}s_{i_k-j}(\boldsymbol{h})\quad  \text{and}\quad \widetilde{\vm}_{j,l}:=\frac{1}{{L}}\overline{d_{l,j}}s_{-j+1}(\boldsymbol{\vz}_l),
\end{equation}
for $l=1,\ldots,L$. 
 
 \textbf{Step 1: Estimate the upper bounds of $\|\boldsymbol{m}_{j,0}\|_2$, $\left\|\sum_{l=1}^L \boldsymbol{m}_{j,l}\right\|_2$, and $\left\|\sum_{l=1}^L \widetilde{\vm}_{j,l}\right\|_2$ for any fixed $j=1,\ldots,n$.}
 
The estimations are  in accordance with Theorem \ref{thm: Bernstein_inequality}, using different selections of  $R$ and $\sigma^2$, respectively:

 \textbf{(i): The estimation of $\|\boldsymbol{m}_{j,0}\|_2$.} 
 
Given that $\mathbb{E}(|d_{l,j}|^2 - 1)=0$, $\max_{l} ||d_{l,j}|^2 - 1| \leq 2\nu^2$, for $l=1,\ldots,L$, and $\sum_{l=1}^L \mathbb{E}(|d_{l,j}|^2 - 1)^2 \leq \nu^2 L$, applying Theorem \ref{thm: Bernstein_inequality}, for any $\epsilon \in (0,1)$, we obtain:
\[
\mathbb{P}\left(\frac{1}{L}\Big|\sum_{l=1}^L\left(|d_{l,j}|^2-1\right)\Big|\geq \epsilon \right)\leq 2\exp\left(-\frac{L\epsilon^2/2}{\nu^2+\nu^2\epsilon/3}\right)\leq 2\exp\left(-\frac{L\epsilon^2}{2\nu^2}\right),
\]
which immediately gives
\begin{equation}\label{eqn: temp_M1}
\mathbb{P}\left(\|\boldsymbol{m}_{j,0}\|_2\geq \epsilon\cdot |x_{j}|\cdot\|\boldsymbol{h}\|_2\right)\leq 2\exp\left(-\frac{L\epsilon^2}{2\nu^2}\right),\quad \text{for  any}\ \epsilon\in (0,1).
\end{equation}

 \textbf{(ii): The estimation of $\left\|\sum_{l=1}^L \boldsymbol{m}_{j,l}\right\|_2$.} 
 
  On one hand, we note that $\mathbb{E}\boldsymbol{m}_{j,l}=\boldsymbol{0}$, for all $l=1,\ldots,L$, and
  \begin{equation}
  \label{eqn: B_h}
\begin{aligned}
&\max_{l}\|\boldsymbol{m}_{j,l}\|_2\\
\leq &\max_{l}\frac{|\boldsymbol{d}_{l,{j}}|}{L}\sqrt{\sum_{k}|d_{l,i_k}|^2\cdot |x_{i_k}|^2\cdot\|\boldsymbol{h}\|_2^2+\sum_{k_1\neq k_2}|x_{i_{k_1}}|\cdot|\overline{x_{i_{k_2}}}|\cdot|d_{l,i_{k_1}}|\cdot|\overline{d_{l,i_{k_2}}}|\cdot|\langle s_{i_{k_1}-j}(\boldsymbol{h}),{s_{i_{k_2}-j}(\boldsymbol{h})}\rangle}|\\
\overset{(a)}\leq &\max_{l}\frac{\nu}{L}\sqrt{\nu^2\|\boldsymbol{h}\|_2^2\|\vx\|_2^2+\nu^2\mu_h\sum_{k_1\neq k_2}|x_{i_{k_1}}|\cdot|{x_{i_{k_2}}}|\cdot\|\boldsymbol{h}\|_2^2}\\
\overset{(b)}\leq & \frac{\nu^2\sqrt{1+\mu_hK}\|\boldsymbol{h}\|_2\|\boldsymbol{x}\|_2}{L}= \frac{\nu^2\sqrt{1+\mu_hK}\|\boldsymbol{h}\|_2}{L}.
\end{aligned}
\end{equation}
Here, inequality $(a)$ follows from the definition of $\mu_h$, while $(b)$ is derived using the bound 
\[
\sum_{k_1 \neq k_2} |x_{i_{k_1}}| \cdot |{x_{i_{k_2}}}| \leq \|\vx\|_1^2 \leq K \|\vx\|_2^2=K,
\]
 under the assumption that $\|\boldsymbol{x}\|_0 \leq K$ and $\|\vx\|_2 = 1$.

On the other hand, we can compute:
\begin{equation}\label{eqn: sigma_h}
\sum_{l=1}^L\mathbb{E}\|\boldsymbol{m}_{j,l}\|_2^2=\frac{1}{L^2}\sum_{l=1}^L\sum_{i_k\neq j}\mathbb{E}|d_{l,j}|^2|d_{l,i_k}|^2|x_{i_k}|^2\|\boldsymbol{h}\|_2^2\leq\frac{1}{L}\|\boldsymbol{h}\|_2^2\cdot\|\vx\|_2^2=\frac{1}{L}\|\boldsymbol{h}\|_2^2.
\end{equation}
Applying Bernstein's inequality in Theorem \ref{thm: Bernstein_inequality}, we obtain the probability bound: for any $\epsilon \in (0,1)$,
\begin{equation}\label{eqn: temp_M2}
\mathbb{P}\left(\Big\|\sum_{l=1}^L\boldsymbol{m}_{j,l}\Big\|_2\geq \epsilon\|\boldsymbol{h}\|_2\right)
\leq  2n\exp\left(-\frac{L\epsilon^2}{2+2\nu^2\sqrt{1+\mu_hK}\epsilon/3}\right)\leq  2n\exp\left(-\frac{L\epsilon^2}{2\nu^2\sqrt{1+\mu_hK}}\right).
\end{equation}

 \textbf{(iii): The estimation of $\left\|\sum_{l=1}^L \widetilde{\vm}_{j,l}\right\|_2$.}   
 
 By direct calculations, we verify that  $\mathbb{E}\widetilde{\vm}_{j,l}=\boldsymbol{0}$, for all $l=1,\ldots,L$. Furthermore, we establish the bounds:
 \[
 \max_{l}|\widetilde{\vm}_{j,l}|\leq \frac{\nu}{{L}}\max_{l}\|\vz_l\|_2\leq \frac{\nu}{{L}}\|\vZ\|_F\quad \text{and}\quad \sum_{l=1}^{L}\mathbb{E}\|\widetilde{\vm}_{j,l}\|_2^2=\frac{1}{L^2}\sum_{l=1}^L\|\vz_{l}\|_2^2\leq \frac{1}{L}\|\vZ\|_F^2.
 \]
 Therefore, for any $\epsilon_1\in (0,1)$, it holds:
 \begin{equation}
 \label{eqn: temp_M3}
 \mathbb{P}\left(\Big\|\sum_{l=1}^L\widetilde{\vm}_{j,l}\Big\|_2\geq \epsilon_1\|\vZ\|_F\right)\leq 2n\exp\left(-\frac{\epsilon_1^2\|\vZ\|^2_F}{\frac{2}{L}\|\vZ\|_F^2+\frac{2\nu\epsilon_1}{3L}\|\vZ\|_F^2}\right)\leq 2n\exp\left(-\frac{{L}\epsilon_1^2}{3\nu}\right).
 \end{equation}
 
 \textbf{Step 2: Taking $j^{\#}:=\text{argmax}_{j}\|\vh_j\|_2$,  estimate $\|\vh_{j^{\#}}\|_2$ and $\left\|e^{-\mathrm{i}\theta_{j^{\#}}}\|\vh\|_2 \vh_{j^{\#}}-\|\vh_{j^{\#}}\|_2\vh\right\|_2$ with $\theta_{j^{\#}}$ satisfying $e^{-\mathrm{i}\theta_{j^{\#}}}\cdot x_{j^{\#}}=|x_{j^{\#}}|$.}

Substituting  (\ref{eqn: temp_M1}), (\ref{eqn: temp_M2}), and (\ref{eqn: temp_M3}) into  (\ref{eqn: h_j}) and applying the union probability bound, we can directly obtain that for any $j = 1, \ldots, n$:
 \begin{equation}\label{eqn: Y_j_new}
  (1-\epsilon)\cdot |x_{j}|\cdot\|\boldsymbol{h}\|_2-\epsilon\|\boldsymbol{h}\|_2-\epsilon_1\|\vZ\|_F\leq \|\vh_j\|_2 \leq (1+\epsilon)\cdot |x_{j}|\cdot\|\boldsymbol{h}\|_2+\epsilon\|\boldsymbol{h}\|_2+\epsilon_1\|\vZ\|_F
  \end{equation}
  and 
  \begin{equation}
  \label{eqn: Y_new1}
  \Big| \|\vh_j\|_2-|x_{j}|\cdot\|\boldsymbol{h}\|_2\Big|\leq \|\boldsymbol{m}_{j,0}\|_2+\left\|\sum_{l=1}^L\boldsymbol{m}_{j,l}\right\|_2+\left\|\sum_{l=1}^L\widetilde{\vm}_{j,l}\right\|_2\leq \epsilon\cdot|x_{j^{\#}}|\cdot\|\boldsymbol{h}\|_2+\epsilon\|\boldsymbol{h}\|_2+\epsilon_1\|\vZ\|_F
  \end{equation}
 with probability at least 
 \begin{equation}
 \label{eqn: probability}
 1-4n^2\exp\left(-\frac{L\epsilon^2}{2\nu^2\sqrt{1+\mu_hK}}\right)-2n^2\exp\left(-\frac{{L}\epsilon_1^2}{3\nu}\right).
 \end{equation}
 
   Since $\max_{j} \|\vh_j\|_2 \leq \|\vh_{j^{\#}}\|_2$,  (\ref{eqn: Y_j_new}) gives us the following lower bound for $\|\vh_{j^{\#}}\|_2$: 
 \begin{equation}
 \label{eqn: upper_h_temp1}
\|\vh_{j^{\#}}\|_2\geq  (1-\epsilon)\cdot \|\vx\|_\infty\|\boldsymbol{h}\|_2-\epsilon\|\boldsymbol{h}\|_2-\epsilon_1\|\vZ\|_F\geq (1-\epsilon)\cdot \|\vx\|_\infty\|\boldsymbol{h}\|_2-\epsilon\|\boldsymbol{h}\|_2-C\epsilon_1\|\vh\|_2,
 \end{equation}
 as $\|\vZ\|_F\leq C\|\vx\|_2\|\vh\|_2$ and $\|\vx\|_2=1$.  
 
 On the other hand,  we also have the following bound for the difference between $e^{-\mathrm{i}\theta_{j^{\#}}}\|\vh\|_2 \vh_{j^{\#}}$ and $\|\vh_{j^{\#}}\|_2\vh$:
 \begin{equation}
 \label{eqn: upper_h_temp2}
 \begin{aligned}
&\left\|e^{-\mathrm{i}\theta_{j^{\#}}}\|\vh\|_2 \vh_{j^{\#}}-\|\vh_{j^{\#}}\|_2\vh\right\|_2\\
=&\left\||x_{j^{\#}}|\cdot\|\vh\|_2\boldsymbol{h}+e^{-\mathrm{i}\theta_{j^{\#}}}\cdot\|\vh\|_2\boldsymbol{m}_{j^{\#},0}+e^{-\mathrm{i}\theta_{j^{\#}}}\cdot\|\vh\|_2\sum_{l=1}^L\boldsymbol{m}_{j^{\#},l}+e^{-\mathrm{i}\theta_{j^{\#}}}\cdot\|\vh\|_2\sum_{l=1}^L\widetilde{\vm}_{j^{\#},l}-\|\vh_{j^{\#}}\|_2\vh\right\|_2\\
\leq &\Big\||x_{j^{\#}}|\cdot\|\vh\|_2\boldsymbol{h}-\|\vh_{j^{\#}}\|_2\vh\Big\|_2+\|\vh\|_2\left(\|\boldsymbol{m}_{j^{\#},0}\|_2+\left\|\sum_{l=1}^L\boldsymbol{m}_{j^{\#},l}\right\|_2+\left\|\sum_{l=1}^L\widetilde{\vm}_{j^{\#},l}\right\|_2\right)\\
\overset{(c)}\leq &2\|\vh\|_2\cdot (\epsilon\cdot |x_{j^{\#}}|\cdot\|\boldsymbol{h}\|_2+\epsilon\|\boldsymbol{h}\|_2+\epsilon_1\|\vZ\|_F)
\leq 2(2\epsilon+C\epsilon_1)\|\vh\|_2^2,
\end{aligned}
 \end{equation}
where inequality $(c)$ follows from (\ref{eqn: Y_new1}). 

\textbf{Step 3: Estimate the distance between $\frac{\boldsymbol{h}_{j^{\#}}}{\|\boldsymbol{h}_{j^{\#}}\|_2}$ and $\frac{\boldsymbol{h}}{\|\boldsymbol{h}\|_2}$.}

 Using (\ref{eqn: upper_h_temp1}) and (\ref{eqn: upper_h_temp2}), we can obtain the following bound:
\begin{equation}\label{eqn: dist_lower_4epsilon}
\begin{split}
\text{dist}\left(\frac{\boldsymbol{h}_{j^{\#}}}{\|\boldsymbol{h}_{j^{\#}}\|_2},\frac{\boldsymbol{h}}{\|\boldsymbol{h}\|_2}\right)\leq& \left\|\frac{e^{-\mathrm{i}\theta_{j^{\#}}}\boldsymbol{h}_{j^{\#}}}{\|\boldsymbol{h}_{j^{\#}}\|_2}-\frac{\boldsymbol{h}}{\|\boldsymbol{h}\|_2}\right\|_2
=\frac{\left\|e^{-\mathrm{i}\theta_{j^{\#}}}\|\vh\|_2 \vh_{j^{\#}}-\|\vh_{j^{\#}}\|_2\vh\right\|_2}{\|\boldsymbol{h}_{j^{\#}}\|_2\|\boldsymbol{h}\|_2}\\
\leq &\frac{2(2\epsilon+C\epsilon_1)\|\vh\|_2^2}{\|\boldsymbol{h}_{j^{\#}}\|_2\|\boldsymbol{h}\|_2}\leq \frac{2(2\epsilon+C\epsilon_1)}{(1-\epsilon)\|\vx\|_{\infty}-\epsilon-C\epsilon_1}.
\end{split}
\end{equation}

Assuming that $\epsilon < \frac{\|\boldsymbol{x}\|_\infty}{4} < \frac{1}{4}$ and $C\epsilon_1< \frac{\|\vx\|_{\infty}}{4}$, we can conclude that 
\[
(1 - \epsilon)\|\vx\|_{\infty} - \epsilon - C\epsilon_1 > \frac{\|\vx\|_\infty}{4}.
\]
Thus, the distance bound simplifies to
\begin{equation}\label{eqn: dist_lower_4epsilon_new}
\text{dist}\left(\frac{\boldsymbol{h}_{j^{\#}}}{\|\boldsymbol{h}_{j^{\#}}\|_2},\frac{\boldsymbol{h}}{\|\boldsymbol{h}\|_2}\right)\leq \frac{8(2\epsilon+C\epsilon_1)}{\|\boldsymbol{x}\|_\infty}.
\end{equation}
Let us define $\widetilde{\epsilon} := \frac{16\epsilon}{\|\vx\|_{\infty}}$ and $\widetilde{\epsilon}_1 := \frac{8\epsilon_1}{\|\vx\|_{\infty}}$, so that $\widetilde{\epsilon} < 4$ and $C\widetilde{\epsilon}_1 < 2$. We can further simplify  (\ref{eqn: dist_lower_4epsilon_new}) to:
\begin{equation}
\label{eqn: final}
\text{dist}\left(\frac{\boldsymbol{h}_{j^{\#}}}{\|\boldsymbol{h}_{j^{\#}}\|_2},\frac{\boldsymbol{h}}{\|\boldsymbol{h}\|_2}\right)\leq \widetilde{\epsilon}+C\widetilde{\epsilon}_1.
\end{equation}
Substituting this result into the probability expression given in (\ref{eqn: probability}), we obtain the following bound:
 \begin{equation}
 1-4n^2\exp\left(-\frac{L\|\vx\|^2_\infty\widetilde{\epsilon}^2}{512\cdot\nu^2\sqrt{1+\mu_hK}}\right)-2n^2\exp\left(-\frac{{L}\|\vx\|^2_\infty\widetilde{\epsilon}_1^2}{192\nu}\right).
 \end{equation}
Thus, for sufficiently large $L$, specifically when
\[
L\gtrsim \frac{\sqrt{1+\mu_h K}\nu^2}{\|\boldsymbol{x}\|_\infty^2\min\{\widetilde{\epsilon}^2,\widetilde{\epsilon}_1^2\}}\log n,
\]
 we conclude that (\ref{eqn: final}) holds with probability at least $1-\frac{1}{n}$. This completes the proof.
\end{proof}

\section{Proof of Theorem \ref{thm: x_estimation}}\label{sec: Thm4}

First and foremost, we present two technical lemmas that play  fundamental roles in the proof of Theorem \ref{thm: x_estimation}. For convenience, in the following statements, rewrite $\vy_l$, for $l = 1, \dots, L$, as
\[
\begin{aligned}
\vy_l=&\vh\circledast(\vd_l\odot \vx)+\vz_l=\vh_0\circledast(\vd_l\odot e^{\mathrm{i}\theta_0}\|\vh\|_2\vx)+ (e^{-\mathrm{i}\theta_0}\vh/\|\vh\|_2-\vh_0)\circledast(\vd_l\odot e^{\mathrm{i}\theta_0}\|\vh\|_2\vx)+\vz_{l},
\end{aligned}
\]
which simplifies to
\[
\vy_l = \A_l \left( e^{\mathrm{i} \theta_0} \|\vh\|_2 \vx \right) + \widetilde{\vz}_l,
\]
where 
\begin{equation}
\label{eqn: A_l and z_l}
\A_{l}(\vz):=\vh_0\circledast(\vd_l\odot \vz),\quad \text{and}\quad \widetilde{\vz}_l:=(e^{-\mathrm{i}\theta_0}\vh/\|\vh\|_2-\vh_0)\circledast(\vd_l\odot e^{\mathrm{i}\theta_0}\|\vh\|_2\vx)+\vz_{l},
\end{equation}
for $l=1,\ldots,L$. 

\begin{lem}\label{lem: RIP}
Let $\boldsymbol{D}_l=\mathrm{diag}(\vd_l)$, for $l = 1, \dots, L$, be independent diagonal matrices, where the diagonal entries are independent copies of $g \in \mathbb{C}$ as defined in Definition \ref{def: g} with parameter $\nu$.  
Consider any fixed vector $\vh_{0} \in \mathbb{S}^{n-1}$ and let $\mathcal{A}_l$ be as defined in the first part of equation (\ref{eqn: A_l and z_l}).
Then, for all $\vz \in \mathbb{C}^n$, it holds that:
\begin{equation}
\label{eqn: RIP}
\frac{4}{5}\|\vz\|_2^2\leq \frac{1}{L}\sum_{l=1}^{L}\|\A_l(\vz)\|_2^2\leq \frac{6}{5}\|\vz\|_2^2,
\end{equation}
with probability at least $1-\frac{1}{n}$, provided that $L\gtrsim {{\mu}_0\nu^2}\log n$. Here $\mu_0=\|\widehat{\vh}_0\|^2_\infty/\|\vh_0\|^2_2=\|\widehat{\vh}_0\|^2_\infty$ is the coherence parameter.
\end{lem}
\begin{proof}
The proof is postponed in Section \ref{sec: RIP}. 
\end{proof}
\begin{lem}\label{lem: noise_upper}
Let $\boldsymbol{D}_l=\mathrm{diag}(\vd_l)$, for $l = 1, \dots, L$, be independent diagonal matrices, where the diagonal entries are independent copies of $g \in \mathbb{C}$ as defined in Definition \ref{def: g} with parameter $\nu$.  
Consider any fixed vector $\vh_{0} \in \mathbb{S}^{n-1}$ and let $\mathcal{A}_l$ and $\widetilde{\vz}_l$, for $l=1,\ldots,L$,  be as defined in (\ref{eqn: A_l and z_l}). Then, we have the following upper bound on the noise term: 
\[
\Big\|\frac{1}{L}\sum_{l=1}^L\mathcal{A}_l^*(\widetilde{\vz}_l)\Big\|_\infty\leq (2\epsilon+C)\|\vh\|_2.
\]
with probability at least $1-\frac{1}{n}$, provided that $L\gtrsim {{\mu}_0\nu^2}\log n$. Here $\mu_0=\|\widehat{\vh}_0\|^2_\infty/\|\vh_0\|^2_2=\|\widehat{\vh}_0\|^2_\infty$, which is the same as in Lemma \ref{lem: RIP}. \end{lem}
\begin{proof}
The proof is postponed in Section \ref{sec: noise_upper}.
\end{proof}
Now we begin to prove Theorem \ref{thm: x_estimation}.
\begin{proof}[Proof of Theorem \ref{thm: x_estimation}]
In \cite[Theorem 3.1]{JL14}, Lin and Li established that if the measurement model $\vy = \vA \vx_0 + \vz$ involves a linear measurement matrix $\vA$ and a noise term $\vz$ satisfying $\|\vA^* \vz\|_\infty \leq \frac{\lambda}{2},$ and the matrix $\vA$ satisfies the following restricted isometry property (RIP)-type condition:\begin{equation}
\label{eqn: RIP_A}
\frac{4}{5}\|\vw\|_2\leq \|\vA\vw\|_2^2\leq \frac{6}{5}\|\vw\|_2^2
\end{equation} 
for all $\vw$, then the solution $\boldsymbol{x}^{\#}$ to the LASSO problem
\begin{equation}
\min_{\widetilde{\boldsymbol{x}}}\ \frac{1}{2}\|\boldsymbol{A}\widetilde{\boldsymbol{x}}-{\boldsymbol{y}}\|_2^2+\lambda\|\widetilde{\boldsymbol{x}}\|_1
\end{equation}
obeys the following error bound:
\[
\|{\boldsymbol{x}}^{\#}-\boldsymbol{x}_0\|_2\leq \min_{1\leq k\leq K}\left(C_1\sqrt{k}\lambda+C_2\frac{\|\boldsymbol{x}-(\boldsymbol{x})_{[k]}\|_1}{\sqrt{k}}\right),
\]
where $C_1$ and $C_2$ are absolute constants. Although the result in \cite{JL14} primarily focuses on the RIP for sparse signals, the condition in \eqref{eqn: RIP_A}, which applies to all signals, is a stronger assumption.

Applying this result to the structured measurement model given in (\ref{eqn: model01}), we define:
\[
\vA\widetilde{\vx}=\frac{1}{\sqrt{L}}\begin{bmatrix}
\mathcal{A}_1(\widetilde{\vx})\\
\vdots\\
\mathcal{A}_{L}(\widetilde{\vx})
\end{bmatrix},
\quad 
\vz=\frac{1}{\sqrt{L}}\begin{bmatrix}
\widetilde{\vz}_1\\
\vdots\\
\widetilde{\vz}_n
\end{bmatrix}
\quad \text{and}\quad \vx_0=e^{\mathrm{i} \theta_0} \|\vh\|_2 \vx.
\]

By Lemma \ref{lem: RIP} and Lemma \ref{lem: noise_upper}, the condition in (\ref{eqn: RIP_A}) holds with probability at least $1 - \frac{1}{n}$, provided that $L \gtrsim \mu_0 \nu^2 \log n.$ Furthermore, the bound on the noise term follows as $\|\vA^* \vz\|_\infty = \Big\|\frac{1}{L} \sum_{l=1}^L \mathcal{A}_l^* (\widetilde{\vz}_l) \Big\|_\infty \leq (2\epsilon + C) \|\vh\|_2$ with probability at least $1 - \frac{1}{n}$, provided that $L \gtrsim \mu_0 \nu^2 \log n.$ 
Setting $\lambda := 2(2\epsilon + C) \|\vh\|_2 = 2(2\epsilon + C) \|\vh\|_2 \|\vx\|_2,$ we can arrive at the desired conclusion.
\end{proof}

\section{Proof of Lemma \ref{lem: RIP}}\label{sec: RIP}
\begin{proof}[Proof of Lemma \ref{lem: RIP}]
First, we directly observe that
\begin{equation}
\label{eqn: c_h_upper}
\|\vC_{\vh_{0}}^{*}  \vC_{\vh_{0}}\|=\max_{\|\boldsymbol{z}\|_2=1}\boldsymbol{z}^{*}  \vC_{\vh_{0}}^{*}  \vC_{\vh_{0}}\boldsymbol{z}\overset{(a)}=\max_{\|\boldsymbol{z}\|_2=1}\frac{1}{n}\|(\boldsymbol{F}\vh_{0})\odot(\boldsymbol{F}\boldsymbol{z})\|_2^2=\max_{\|\boldsymbol{z}\|_2=1}\|(\boldsymbol{F}\vh_{0})\odot\boldsymbol{z}\|_2^2=\|\boldsymbol{F}\vh_{0}\|_\infty^2=\mu_0.
\end{equation}
Here, $(a)$ follows from the identity $\boldsymbol{F}(\vC_{\vh_0} \boldsymbol{z}) = (\boldsymbol{F} \vh_0) \odot (\boldsymbol{F} \boldsymbol{z})$ and the fact that $n \|\vC_{\vh_0} \boldsymbol{z}\|_2^2 = \|\boldsymbol{F} (\vC_{\vh_0} \boldsymbol{z})\|_2^2$.  

Let $\boldsymbol{M}_l := \frac{1}{L}\left(\overline{\boldsymbol{D}_{l}} \vC_{\vh_{0}}^{*} \vC_{\vh_{0}} \boldsymbol{D}_{l} - \boldsymbol{I}\right)$, for $l = 1, \ldots, L$. A direct calculation shows that proving (\ref{eqn: RIP}) is equivalent to demonstrating the following bound:
\begin{equation}
\label{eqn: M_upper}
\Big\|\sum_{l=1}^{L} \boldsymbol{M}_l\Big\|\leq \frac{1}{5}.
\end{equation}
Since $\|\vh_{0}\|_2 = 1$, it follows that
\[
\mathbb{E} \boldsymbol{M}_l=\frac{1}{L}\left(\text{diag}(\vC_{\vh_{0}}^{*}  \vC_{\vh_{0}})-\boldsymbol{I}\right)=\boldsymbol{0}.
\]
Next, for $l = 1, \ldots, L$, we obtain the following upper bound for $\|\boldsymbol{M}_l\|$:
\begin{equation}\label{eqn: mid1}
\begin{split}
\|\boldsymbol{M}_l\|
=&\left\|\frac{1}{L}\overline{\boldsymbol{D}_{l}}\vC_{\vh_{0}}^{*}  \vC_{\vh_{0}}\boldsymbol{D}_l-\frac{1}{L}\boldsymbol{I}\right\|\leq \left\|\frac{1}{L}\overline{\boldsymbol{D}_{l}}\vC_{\vh_{0}}^{*}  \vC_{\vh_{0}}\boldsymbol{D}_{l}\right\|+\frac{1}{L}\leq\frac{\nu^2}{L}\|\vC_{\vh_{0}}^{*}  \vC_{\vh_{0}}\|+\frac{1}{L}\\
\overset{(b)}\leq & \frac{\mu_0\nu^2+1}{L}\overset{(c)}\leq \frac{2\mu_0\nu^2}{L}.
\end{split}
\end{equation}
Here $(b)$ follows from (\ref{eqn: c_h_upper}), and $(c)$ is derived from
 $
 \mu_0=\|\boldsymbol{F}\vh_{0}\|_\infty^2\geq \frac{1}{n}\|\boldsymbol{F}\vh_{0}\|_2^2=\|\vh_{0}\|_2^2=1.
 $

Besides, we have
\begin{equation}\label{eqn: mid2}
\begin{split}
\Big\|\sum_{l=1}^{L}\mathbb{E}(\boldsymbol{M}_l\boldsymbol{M}_l^{*} )\Big\|=&\Big\|\sum_{l=1}^{L}\mathbb{E}(\boldsymbol{M}_l^{*} \boldsymbol{M}_l)\Big\|\\
=&\Big\|\frac{1}{L^2}\sum_{l=1}^{L}\mathbb{E}\left(\left(\overline{\boldsymbol{D}_{l}}\vC_{\vh_{0}}^{*}  \vC_{\vh_{0}}\boldsymbol{D}_{l}-\boldsymbol{I}\right)\left(\overline{\boldsymbol{D}_{l}}\vC_{\vh_{0}}^{*}  \vC_{\vh_{0}}\boldsymbol{D}_{l}-\boldsymbol{I}\right)\right)\Big\|\\
=&\Big\|\frac{1}{L^2}\sum_{l=1}^{L}\mathbb{E}\left(\overline{\boldsymbol{D}_{l}}\vC_{\vh_{0}}^{*}  \vC_{\vh_{0}}{\boldsymbol{D}_{l}} \overline{\boldsymbol{D}_{l}}  \vC_{\vh_{0}}^{*}  \vC_{\vh_{0}}\boldsymbol{D}_{l}-2\overline{\boldsymbol{D}_{l}}\vC_{\vh_{0}}^{*}  \vC_{\vh_{0}}\boldsymbol{D}_{l}+\boldsymbol{I}\right)\Big\|\\
\overset{(d)}\leq &\frac{\nu^2}{L}\left\|\text{diag}(\vC_{\vh_{0}}^{*}  \vC_{\vh_{0}}\vC_{\vh_{0}}^{*}  \vC_{\vh_{0}})\right\|+\frac{1}{L}\leq \frac{\nu^2}{L}\left\| \vC_{\vh_{0}}\vC_{\vh_{0}}^{*}\right\|+\frac{1}{L}\\
\overset{(e)}\leq & \frac{\mu_0\nu^2+1}{L}\leq \frac{2\mu_0\nu^2}{L}.
\end{split}
\end{equation}
Here $(d)$ relies on the fact that 
\[
\mathbb{E}\left(\overline{\boldsymbol{D}_{l}}\vC_{\vh_{0}}^{*}  \vC_{\vh_{0}}{\boldsymbol{D}_{l}} \overline{\boldsymbol{D}_{l}}  \vC_{\vh_{0}}^{*}  \vC_{\vh_{0}}\boldsymbol{D}_{l}\right)\preceq \nu^2\mathbb{E}\left(\overline{\boldsymbol{D}_{l}}\vC_{\vh_{0}}^{*}  \vC_{\vh_{0}} \vC_{\vh_{0}}^{*}  \vC_{\vh_{0}}\boldsymbol{D}_{l}\right)=\nu^2\text{diag}(\vC_{\vh_{0}}^{*}  \vC_{\vh_{0}}\vC_{\vh_{0}}^{*}  \vC_{\vh_{0}}),
\]
and $(e)$ also relies on  (\ref{eqn: c_h_upper}). 

Next, applying Bernstein's inequality from Theorem \ref{thm: Bernstein_inequality} with the sequence $\{\boldsymbol{M}_k\}$, and setting $R := \frac{2 \mu_0 \nu^2}{L}$  and $\sigma^2 := \frac{2 \mu_0 \nu^2}{L}$ as per (\ref{eqn: mid1}) and (\ref{eqn: mid2}), we obtain:
\begin{equation}
\label{eqn: berstein_RIP}
\mathbb{P}\left\{ \Big\|\sum_{k}\boldsymbol{M}_{k}\Big\|\geq t\right\} \leq 2n\exp\left(\frac{-Lt^{2}/2}{2\mu_0\nu^2+2\mu_0\nu^2t/3}\right).
\end{equation}
By setting $t = \frac{1}{4}$ and $L \gtrsim \mu_0 \nu^2 \log n$ in (\ref{eqn: berstein_RIP}), we obtain
\[
\Big\|\sum_{l=1}^{L} \boldsymbol{M}_l\Big\|\leq \frac{1}{5}
\]
with probability at least $1-\frac{1}{n}$. This concludes the proof of the bound in (\ref{eqn: M_upper}).
 \end{proof}

\section{Proof of Lemma \ref{lem: noise_upper}}\label{sec: noise_upper}
\begin{proof}[The proof of Lemma \ref{lem: noise_upper}]
By direct calculations, we obtain:
\[
\frac{1}{L}\sum_{l=1}^L\mathcal{A}_l^*(\widetilde{\vz}_l)=\frac{1}{L}\sum_{l=1}^L e^{\mathrm{i}\theta_{0}}\|\boldsymbol{h}\|_2\overline{\boldsymbol{D}_{l}}\vC_{\vh_{0}}^{*} (e^{-\mathrm{i}\theta_{0}}\vC_{\boldsymbol{h}/\|\boldsymbol{h}\|_2}-\vC_{\vh_{0}})\boldsymbol{D}_{l}\boldsymbol{x}+\frac{1}{L}\sum_{l=1}^L\overline{\vD_{l}}\vC^*_{\vh_0}\vz_l.
\]
Let $\vd_l = [d_{l,1}, \ldots, d_{l,n}]^{T}$ for $l = 1, \ldots, L$, and denote $s_{\tau}(\vh_0)$ for $\tau \in \{0, \ldots, n-1\}$ as specified in equation (\ref{eqn: s_tau}). We express $\frac{1}{L}\sum_{l=1}^L\mathcal{A}_l^*(\widetilde{\vz}_l)$ in component-wise form as:
\[
\frac{1}{L}\sum_{l=1}^L\mathcal{A}_l^*(\widetilde{\vz}_l)=[u_1,\ldots,u_n]^{T}.
\] 
For each $j = 1, \ldots, n$, we decompose $u_j$  as follows:
\[
\begin{aligned}
u_j=\sum_{l=1}^L \alpha_{j,l}+\sum_{l=1}^{L}\beta_{j,l},
\end{aligned}
\]
where the terms $\alpha_{j,l}$ and $\beta_{j,l}$ are defined as:
\[
\alpha_{j,l}:=\frac{1}{L}e^{\mathrm{i}\theta_{0}}\|\boldsymbol{h}\|_2\overline{d_{l,j}}s_{-j+1}^*(\vh_0) (e^{-\mathrm{i}\theta_{0}}\vC_{\boldsymbol{h}/\|\boldsymbol{h}\|_2}-\vC_{\vh_{0}})\boldsymbol{D}_{l}\boldsymbol{x},\qquad \text{and}\qquad \beta_{j,l}:=\frac{1}{L}\overline{d_{l,j}}s_{-j+1}^*(\vh_0)\vz_l.
\]
To establish an upper bound for $\Big\|\frac{1}{L}\sum_{l=1}^L\mathcal{A}_l^*(\widetilde{\vz}_l)\Big\|_\infty$, we first analyze the upper bounds of $\Big|\sum_{l=1}^L\alpha_{j,l}\Big|$  and $\Big|\sum_{l=1}^L\beta_{j,l}\Big|$ for any fixed $j = 1, \ldots, n$:

\textbf{(1): Estimation of the upper bound of $|\sum_{l=1}^L\alpha_{j,l}|$.}

We begin by noting that
\[
\begin{aligned}
|\mathbb{E}(\alpha_{j,l})|=&\frac{1}{L}\|\vh\|_2\cdot|x_{j}|\cdot|s^*_{-j+1}(\vh_0)(e^{-\mathrm{i}\theta_0}s_{-j+1}(\vh/\|\vh\|_2)-s_{-j+1}(\vh_0))|\\
=&\frac{1}{L}\|\vh\|_2\cdot|x_{j}|\cdot\|s^*_{-j+1}(\vh_0)\|_2\cdot \|(e^{-\mathrm{i}\theta_0}s_{-j+1}(\vh/\|\vh\|_2)-s_{-j+1}(\vh_0))\|_2\\
\leq& \frac{\epsilon}{L}\|\vh\|_2,
\end{aligned}
\]
and 
\[
|\alpha_{j,l}|\leq \frac{1}{L}\nu\|\vh\|_2\cdot \|s_{-j+1}^*(\vh_0) (e^{-\mathrm{i}\theta_{0}}\vC_{\boldsymbol{h}/\|\boldsymbol{h}\|_2}-\vC_{\vh_{0}})\|_2\cdot\|\vD_l\vx\|_2\leq \frac{\epsilon\sqrt{\mu_0}\nu^2}{L}\|\vh\|_2,
\]
where the inequality follows from $\| e^{-\mathrm{i}\theta_0} \vh/\|\vh\|_2 - \vh_0 \|_2 \leq \epsilon$, $\|\vh_0\|_2 = \|\vx\|_2 = 1$, and the following relation:
\begin{equation}
\label{eqn: s_j_temp}
\begin{aligned}
 \|s_{-j+1}^*(\vh_0) (e^{-\mathrm{i}\theta_{0}}\vC_{\boldsymbol{h}/\|\boldsymbol{h}\|_2}-\vC_{\vh_{0}})\|_2=&\|\vC^*_{\widetilde{\vh}}s_{-j+1}(\vh_0)\|_2=\left\|(\frac{1}{n}\vF^*\text{diag}(\vF\widetilde{\vh})\vF)^*s_{-j+1}(\vh_0)\right\|_2\\
 =&\left\|\frac{1}{n}\vF^*\overline{\text{diag}(\vF\widetilde{\vh})}\vF s_{-j+1}(\vh_0)\right\|_2=\left\|\frac{1}{n}\vF^*{\text{diag}(\vF s_{-j+1}(\vh_0))}\overline{\vF\widetilde{\vh}}\right\|_2\\
 \leq &\frac{1}{n}\|\vF^*\|\cdot \|\overline{\vF}\|\cdot \|\widetilde{\vh}\|_2\cdot \|\vF s_{-j+1}(\vh_0)\|_{\infty}\leq \epsilon \sqrt{\mu_0}
 \end{aligned}
\end{equation}
with $\widetilde{\vh}:=e^{-\mathrm{i}\theta_0}\vh/\|\vh\|_2-\vh_0$. The final line above is based on the assumption that $\|\widetilde{\vh}\|_2 \leq \epsilon$, and $\|\vF s_{-j+1}(\vh_0)\|_{\infty}^2 = \|\vF(\vh_0)\|_{\infty}^2 = \mu_0 \geq 1$. 

Furthermore, since $\|\vx\|_\infty \leq \|\vx\|_2 = 1$, we obtain the following:
\[
\begin{aligned}
\sum_{l=1}^{L}\mathbb{E}|\alpha_{j,l}-\mathbb{E}\alpha_{j,l}|^2\leq &\sum_{l=1}^{L}\mathbb{E}|\alpha_{j,l}|^2\\
\leq& \frac{\nu^2}{L}\|\vh\|_2^2\cdot\mathbb{E}|s_{-j+1}^*(\vh_0) (e^{-\mathrm{i}\theta_{0}}\vC_{\boldsymbol{h}/\|\boldsymbol{h}\|_2}-\vC_{\vh_{0}})\boldsymbol{D}_{l}\boldsymbol{x}|^2\\
= &\frac{\nu^2}{L}\|\vh\|_2^2\cdot \|s_{-j+1}^*(\vh_0) (e^{-\mathrm{i}\theta_{0}}\vC_{\boldsymbol{h}/\|\boldsymbol{h}\|_2}-\vC_{\vh_{0}})\|_2^2\cdot \|\vx\|_2^2\\
\leq & \frac{\mu_0\epsilon^2\nu^2}{L}\|\vh\|_2^2.
\end{aligned}
\]
The last inequality also follows from (\ref{eqn: s_j_temp}). 

Therefore, we can apply Theorem \ref{thm: Bernstein_inequality} to the sequence $\{ \alpha_{j,l} - \mathbb{E}[\alpha_{j,l}] \}_{l=1}^{L}$, where $R := \frac{2 \epsilon \sqrt{\mu_0} \nu^2\|\vh\|_2}{L}$ and $\sigma^2 := \frac{\mu_0 \epsilon^2 \nu^2}{L} \|\vh\|_2^2$, yielding the following bound for all $t \geq 0$:
\[
\mathbb{P}\left\{ \Big|\sum_{l=1}^{L}(\alpha_{j,l} - \mathbb{E}[\alpha_{j,l}])\Big|\geq t\right\} \leq 2\exp\left(\frac{-Lt^{2}/2}{\mu_0 \epsilon^2 \nu^2\|\vh\|_2^2+2\epsilon \sqrt{\mu_0} \nu^2\|\vh\|_2t/3}\right).
\]
Taking $t=\epsilon\|\vh\|_2$, it directly leads to 
\begin{equation}
\label{eqn: alpha}
\left|\sum_{l=1}^L\alpha_{j,l}\right|\leq \epsilon\|\vh\|_2+\sum_{l=1}^L|\mathbb{E}\alpha_{j,l}|\leq 2\epsilon\|\vh\|_2
\end{equation}
with probability at least 
\[
1-2\exp\left(-\frac{L}{4\mu_0\nu^2}\right).
\]

\textbf{(2): Estimation of the upper bound of $|\sum_{l=1}^L\beta_{j,l}|$.}

Here we have $\mathbb{E}\beta_{j,l}=0$ with 
\[
|\beta_{j,l}|\leq \frac{\nu}{L}\|\vz_l\|_2\leq \frac{\nu}{L}\|\vZ\|_F\quad \text{and}\quad \sum_{l=1}^L\mathbb{E}|\beta_{j,l}|^2=\frac{1}{L}\sum_{l=1}^{L}\|\vz_{l}\|_2^2=\frac{1}{L}\|\vZ\|_F^2. 
\]
Applying Theorem \ref{thm: Bernstein_inequality} to the sequence $\{ \beta_{j,l} \}_{l=1}^L$, with $R := \frac{\nu}{L} \|\vZ\|_F$ and $\sigma^2 := \frac{1}{L} \|\vZ\|_F^2$. The inequality from the theorem states that for all $t \geq 0$,
\[
\mathbb{P}\left\{ \big|\sum_{l=1}^{L}\beta_{j,l} \big|\geq t\right\} \leq 2\exp\left(\frac{-Lt^{2}/2}{\|\vZ\|^2_F+\nu\|\vZ\|_Ft/3}\right).
\]
Taking $t = \|\vZ\|_F$, and recalling the assumption that $\|\vZ\|_F \leq C \|\vx\|_2 \|\vh\|_2 = C \|\vh\|_2$, we directly obtain
\begin{equation}
\label{eqn: beta}
\Big|\sum_{l=1}^{L}\beta_{j,l} \Big|\leq \|\vZ\|_F\leq C\|\vh\|_2
\end{equation}
with probability at least 
\[
1-2\exp\left(-\frac{L}{3\nu}\right).
\]

Now we aim to estimate the upper bound of $|u_j|$, for each $j=1,\ldots,n$. 
Using the results from (\ref{eqn: alpha}) and (\ref{eqn: beta}), we can conclude that for any fixed $j = 1, \ldots, n$,
\[
\begin{aligned}
|u_j|\leq \Big|\sum_{l=1}^L \alpha_{j,l}\Big|+\Big|\sum_{l=1}^{L}\beta_{j,l}\Big|
\leq  (2\epsilon+C)\|\vh\|_2.
\end{aligned}
\]
This holds with probability at least
\[
1-2\exp\left(-\frac{L}{4\mu_0\nu^2}\right)-2\exp\left(-\frac{L}{3\nu}\right).
\]

Finally, by applying the union probability bound, and assuming $L \gtrsim \mu_0 \nu^2 \log n$, we conclude that
\[
\Big\|\frac{1}{L}\sum_{l=1}^L\mathcal{A}_l^*(\widetilde{\vz}_l)\Big\|_\infty=\max_{j}|u_j|_{\infty}\leq (2\epsilon+C)\|\vh\|_2,
\]
with probability at least $1-\frac{1}{n}$.

\end{proof}

\bibliographystyle{plain}
\bibliography{references}
\end{document}